\documentclass[lettersize,journal]{IEEEtran}
\usepackage{amsmath}
\usepackage{amssymb}
\usepackage{amsthm}
\usepackage{adjustbox}
\usepackage{caption}
\usepackage{color}
\usepackage[numbers,sort&compress]{natbib}
\usepackage{comment}
\usepackage{enumerate}
\usepackage[colorlinks=true,linkcolor=blue,citecolor=blue,]{hyperref}
\usepackage{graphicx}
\usepackage{dsfont}
\usepackage{subcaption}
\usepackage{tabularray}
\usepackage{tabularx}
\usepackage{setspace}
\usepackage{makecell}

\usepackage[T1]{fontenc}
\usepackage{indentfirst}
\usepackage{graphicx}
\usepackage{gensymb}
\usepackage[ruled, linesnumbered, vlined]{algorithm2e}
\usepackage{stfloats}
\usepackage{titlesec}

\graphicspath{{images/}}
\DeclareCaptionLabelSeparator{custom}{.\quad}
\theoremstyle{definition}

\newtheorem{theorem}{Theorem}
\newtheorem{lemma}{Lemma}
\makeatletter
\renewenvironment{proof}[1][\proofname]{\par
  \pushQED{\qed}%
  \normalfont \topsep6\p@\@plus6\p@\relax
  \trivlist
  \item[\hskip\labelsep
        \itshape
    #1\@addpunct{:}]\ignorespaces
}{%
  \popQED\endtrivlist\@endpefalse
}
\makeatother

\renewcommand{\thesubsubsection}{\Alph{subsection}.\arabic{subsubsection}}

\begin{document}
\title{Optimal 3D Directional WPT Charging via UAV for 3D Wireless Rechargeable Sensor Networks}

	\author{
		Zhenguo~Gao*,~\IEEEmembership{Senior Member,~IEEE}, Hui Li, Yiqin Chen, Qingyu Gao, Zhufang Kuang,~\IEEEmembership{Member,~IEEE}, Shih-Hau Fang,~\IEEEmembership{Scenior Member,~IEEE}, Hsiao-Chun Wu*,~\IEEEmembership{Fellow,~IEEE}
		
		\IEEEcompsocitemizethanks{
			\IEEEcompsocthanksitem Z.~G.~Gao, H.~Li, Y.~Q.~Chen, and Q.~Y.~Gao are with both the Department of Computer Science and Technology in Huaqiao University, and Key Laboratory of Computer Vision and Machine Learning(Huaqiao University), Fujian Province University, Xiamen, FJ, 361021, CHINA. (e-mail: {\tt \{gzg,chwang\}@hqu.edu.cn}); Z.~Kuang is with the School of Computer and Information Engineering, Central South University of Forestry and Technology, Changsha 410004, China (e-mail: {\tt zfkuangcn@163.com}); Shih-Hau Fang is with the Department of Electrical Engineering, Yuan Ze University, Taoyuan City 320, Taiwan. (e-mail:{\tt shfang@saturn.yzu.edu.tw}); H.-C.~Wu is with the Innovation Center for AI Applications, Yuan Ze University, Chungli 32003, Taiwan (e-mail: {\tt eceprofessorwu@gmail.com}).
		}
		
		\thanks{This work was jointly supported by Natural Science Foundation of China under Grants 62372190 and 62072477, and also in part by National Science Foundation, USA under Award Number 2335150.}
	}

%
%
% The paper headers
\markboth{IEEE}%
{\MakeLowercase{\textit{et al. }}: }
\maketitle
\IEEEtitleabstractindextext{%
\begin{abstract}
The high mobility and flexible deployment capability of UAVs make them an impressive option for charging nodes in Wireless Rechargeable Sensor Networks (WRSNs) using Directional Wireless Power Transfer (WPT) technology. However, existing studies largely focus on 2D-WRSNs, lacking designs catering to real 3D-WRSNs. The spatial distribution characteristics of nodes in a 3D-WRSN further increase the complexity of the charging scheduling task, thus requiring a systematic framework to solve this problem. In this paper, we investigated the Directional UAV Charging Scheduling problem for 3D-WRSNs (DCS-3D) and established its NP-hard property, and then proposed a three-step framework named as directional charging scheduling algorithm using Functional Equivalent (FuncEqv) direction set and Lin-Kernighan heuristic (LKH) for 3D-WRSNs (FELKH-3D) to solve it. In FELKH-3D, the challenge of infinite charging direction space is solved by designing an algorithm generating a minimum-size direction set guaranteed to be FuncEqv to the infinite set of whole sphere surface, and the optimaility of the method was proved.To determine the optimal charging tour for the UAV, the LKH algorithm is employed.Simulation experiments demonstrated the superiority of FELKH-3D over other classical algorithms.

\end{abstract}

\begin{IEEEkeywords}
\textbf{ 3D Wireless rechargeable sensor networks, charging schedule, UAV charger, directional WPT.}
\end{IEEEkeywords}
}
\maketitle

\IEEEdisplaynontitleabstractindextext

\IEEEpeerreviewmaketitle

%\IEEEraisesectionheading{}

\section{ Introduction}
\label{sec_intro}

\IEEEPARstart{W}{ireless} Sensor Networks (WSNs) have broad application prospects in many fields such as precision agriculture~\cite{Pradeep} and logistics monitoring~\cite{Jingjing}. For the significant time and maintenance costs of replacing batteries~\cite{Nowrozian}, however, the limited battery capacity of nodes has continuously been a critical bottleneck hindering the prevalence of WSNs~\cite{Huan2020ToC}. 
With the breakthroughs in Wireless Power Transfer (WPT) technology, using Mobile Chargers (MCs) equipped with WPT devices to charge the nodes in Wireless Rechargeable Sensor Networks (WRSNs) has shown great potential in extending network lifespan~\cite{Sikeridis}. Recently, charging scheduling for WRSNs in complex scenarios, where the movement of MCs are restricted, are attracting research efforts. With strong mobility and flexible deployment capabilities~\cite{Xu2023TWC,Kuang2024JIOT,Zhao2024TVTEarly,Zhang2024TMCEarly}, Unmanned Aerial Vehicles (UAVs) are considered more suitable for charging nodes in complex scenarios, and hence utilizing UAVs as flying chargers to charge WRSNs is becoming a hot topic~\cite{liu2024IOJ}. 

Many related research efforts are targeted for WRSNs with nodes distributed on a two-dimensional plane. With the adoption of UAVs in WRSNs, many work begin to focus on WRSNs with nodes distributed in Three-Dimensional (3D) space. For example, the work in ~\cite{Wang} assumed a two-layer network where chargers are to be deployed in an upper plane to charge the nodes in a bottom plane. Nodes are positioned on rooftops and curved faces. Targeted for scenarios where nodes are on curved terrestrial ground, the author of\cite{Lin2021TMC} investigated the UAV charging schedule problem, assuming that the UAV can only radiate energy in the downward half sphere, not fully exploiting the flexibility of UAV chargers. The network there is not full 3D.

In the true 3D network, sensor nodes are spatially distributed across different locations within a volumetric space—for example, mounted on balloons, attached to trees or poles, or deployed in underwater environments—forming a network topology with distinct spatial hierarchies. Leveraging the high mobility and operational flexibility of UAVs, energy replenishment for these spatially dispersed nodes can be effectively achieved. However, due to the layered and three-dimensional distribution of nodes, determining the optimal direction for wireless energy transmission becomes far more complex than in two-dimensional networks, thereby posing significant challenges to directional energy transfer.

This paper aims to systematically solve the UAV charging scheduling problem caused by the spatial distribution characteristics of nodes in a 3D-WRSN, and to make full use of the UAV's WPT capability to improve scheduling efficiency.
We focus on a scenario where a UAV with directional charging ability is used to charge nodes distributed within a 3D region.
The underlying optimization problem is defined as the Directional UAV Charging Scheduling problem for 3D-WRSNs (DCS-3D). Given the rapid decay of wireless power transfer efficiency with distance—and the extremely low efficiency beyond a certain threshold—we, in line with existing literature~\cite{tormar2021TMC}, consider only node locations as candidate charging positions. Hence, the terms node position and charging position are used interchangeably throughout this paper.

We first prove that the DCS-3D problem is NP-hard and then propose a three-step approximation framework: (1) select charging directions at each charging position (i.e., node location), (2) determine the charging time along each selected direction, and (3) construct a round-trip tour visiting all charging positions. A complete solution is derived from these decisions.
\textit{To address the challenge of the infinite charging direction space}, we extend the 2D analysis in~\cite{Gao2024TMC} to 3D and design an algorithm to compute a minimal-size charging direction set that is functionally equivalent (FuncEqv) to the full spherical space. Here, FuncEqv implies that using this reduced set as the directional search space yields optimal charging schedules identical in quality to those derived from the original continuous direction space. This algorithm, named Creating Minimal FuncEqv Direction Set (cMFEDS), is proven to be optimal.
\textit{To address the UAV path planning problem}, we adopt the Lin-Kernighan heuristic (LKH)\cite{lkh_2019}, a state-of-the-art TSP solver known for high solution quality and efficiency. Finally, we integrate the above components into a unified algorithm—FELKH-3D (FuncEqv-direction-based Energy-efficient LKH algorithm for 3D-WRSNs)—as illustrated in Fig.\ref{fig_DCSBGS}. 

\begin{figure}
    \centering
    \includegraphics[width=0.48\textwidth]{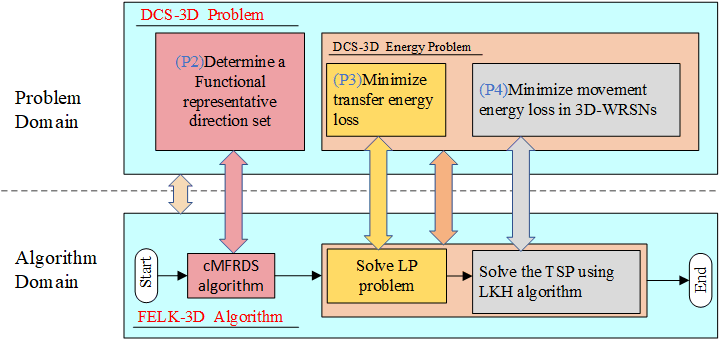}
    \caption{Outline of the FELKH-3D algorithm.}
    \label{fig_DCSBGS}
    \vspace{-0.6cm}
\end{figure}

Main contributions in this paper are summarized as follows:	
\begin{enumerate}[1)]
\item The cMFEDS algorithm was proposed which generates minimum-size charging direction set that is FuncEqv to the original infinite set of the whole sphere surface. cMFEDS's optimality was established.
\item The DCS-3D problem was shown to be NP-hard, and the FELKH-3D algorithm was proposed, which solves the DCS-3D problem in three steps by integrating established methods including cMFEDS and LKH\cite{lkh_2019}.
\item Extensive simulation experiments were performed, which demonstrate the superiority of FELKH-3D.
\end{enumerate}

The remainder of this paper is organized as follows. Sec.~\ref{sec_related_work} reviews related work. Sec.~\ref{sec_models} presents preliminary models. Sec.~\ref{sec_problem_formulation} describes and formulates the DCS-3D problem. Sec.~\ref{sec_address_direction_challenge} provides the analyses and design of the cMFEDS algorithm. Sec.~\ref{sec_algorithm} describes the FELKH-3D algorithm. 
Sec.~\ref{sec_testbed} introduces our testbed experiments.
Sec.~\ref{sec_simulation} presents and analyzes the simulation results. 
Sec.~\ref{sec_conclusion} concludes the paper.

\section{ Related Work}
\label{sec_related_work}
We briefly review researches related to the key components of WPT charging scheduling topic: charging position generation, charging direction selection, and charging tour determination.

\subsection{Charging Position Generation for WPT Chargers}
Despite advances in WPT distance extension, signal attenuation remains a fundamental limitation, making node positions the default charging candidates in most studies. Recent work has explored various optimization approaches: Jiang et al.~\cite{Jiang} proposed grid-based position refinement for cone-shaped chargers in roof-to-ground scenarios, while Liang's team~\cite{Liang2023TVT,LiangShuang} developed clustering methods for 2D-WRSNs and an improved firefly algorithm for UAV deployment. Alternative strategies employ spatial discretization to manage complexity, including piece-wise linearization of energy transfer models~\cite{Dai2021TMC,Dai2017ACM} and spatiotemporal discretization for pseudo-3D uneven terrains~\cite{Lin2021TMC}. These approaches collectively address the trade-off between charging efficiency and computational complexity in position selection.
Although the WPT charging distance had been extended considerably, the energy transfer coefficient is still low due to the serious attenuation of radio signals over distances. As a result, using the positions of the nodes as candidate charging positions is commonly adopted in the literature.

\subsection{Charging Direction Selection for Directional Chargers}

A common approach in existing literature is to generate dedicated charging directions for individual nodes~\cite{lin2017TMC,kumar2021efficient}. Alternatively, fixed direction sets are employed in some works, such as~\cite{Liang2023TVT}, where predefined directions are used irrespective of the spatial distribution of nodes.
More refined methods have also been explored. In~\cite{Jiang}, the authors investigate a pseudo-3D WRSN scenario, where directional chargers are deployed on a roof plane to charge ground-level nodes. Two charging direction selection strategies are introduced: Greedy Cone Covering (GCC) and Adaptive Cone Covering (ACC). GCC iteratively generates directions based on node pairs; depending on whether both nodes can be covered by a single cone, either two or four directions are generated. ACC further compresses the direction set by greedily merging adjacent directions.

To minimize charging delay,~\cite{Lin2021ACM} proposed the Charging Power Discretization (CPD) algorithm, which identifies potential directions from each charging position to nodes within range, then applies k-means clustering to refine the direction set.
In a previous work~\cite{Gao2024TMC}, an optimal charging direction selection method was developed for DMCs in 2D-WRSNs. The concept of functional equivalence (FuncEqv) was introduced, where two direction sets are considered equivalent if they support optimal schedules of equal quality. An optimal algorithm was proposed to generate a minimum-size FuncEqv direction set over the continuous space $[0, 2\pi)$. A similar idea was adopted in~\cite{Dai2022TMC}, though all aforementioned methods are restricted to 2D scenarios and cannot be directly applied to 3D-WRSNs.

\subsection{Charging Tour Searching Methods for Mobile Chargers}

Traditional and heuristic TSP algorithms are widely used for charging tour construction. Gao et al.\cite{Gao2023IOT,Gao2024TMC} employed the Christofides algorithm to determine efficient charging tours in energy-redistribution-assisted WRSNs. 
Chen et al.~\cite{Gao2023Adhoc} extend this to multi-MC, multi-base station settings with enhanced ACO, jointly optimizing energy loss and time span.
Liang et al.~\cite{Liang2023TVT} reduce TSP complexity by strategically grouping sensors and then apply an improved Ant Colony Optimization (ACO) to define fewer stopping points, thereby cutting travel cost.

\section{Preliminary Models}
\label{sec_models}

\subsection{System Model}
\label{s3_system_model}
We consider a 3D-WRSN  contains $n$  nodes $\mathcal{U} = \{u_1,\ldots,u_n\}$ located at positions $\mathcal{L} = \{l_i = (x_i,y_i,z_i) \in \mathbb{R}^3\}$, a UAV equipped with directional wireless power transfer capability, and a base station at location $l_0$. Each node $u_i$ has battery capacity $e_\text{U}(i)$, initial energy level $e_\text{B}(i)$, and energy demand $e_\text{D}(i)$. The charging requirement is satisfied when the condition $\textbf{e}_\textbf{F}{:=}\min\{\textbf{e}_\textbf{B} + \textbf{e}_\textbf{R}, \textbf{e}_\textbf{U}\} \geq (\textbf{e}_\textbf{D} + \textbf{e}_\text{B})$ holds, where $\textbf{e}_\textbf{R}$ represents the received energy vector.
The UAV operates under the following constraints: it begins each mission with full battery capacity $e_\text{U}(0)$ at the base station, maintains a constant flight speed $\bar{v}$, carries sufficient energy to complete all assigned charging tasks, and only engages in wireless power transmission when positioned at designated node locations.

\begin{figure}[ht]
\centering
\includegraphics[width=0.4\textwidth]{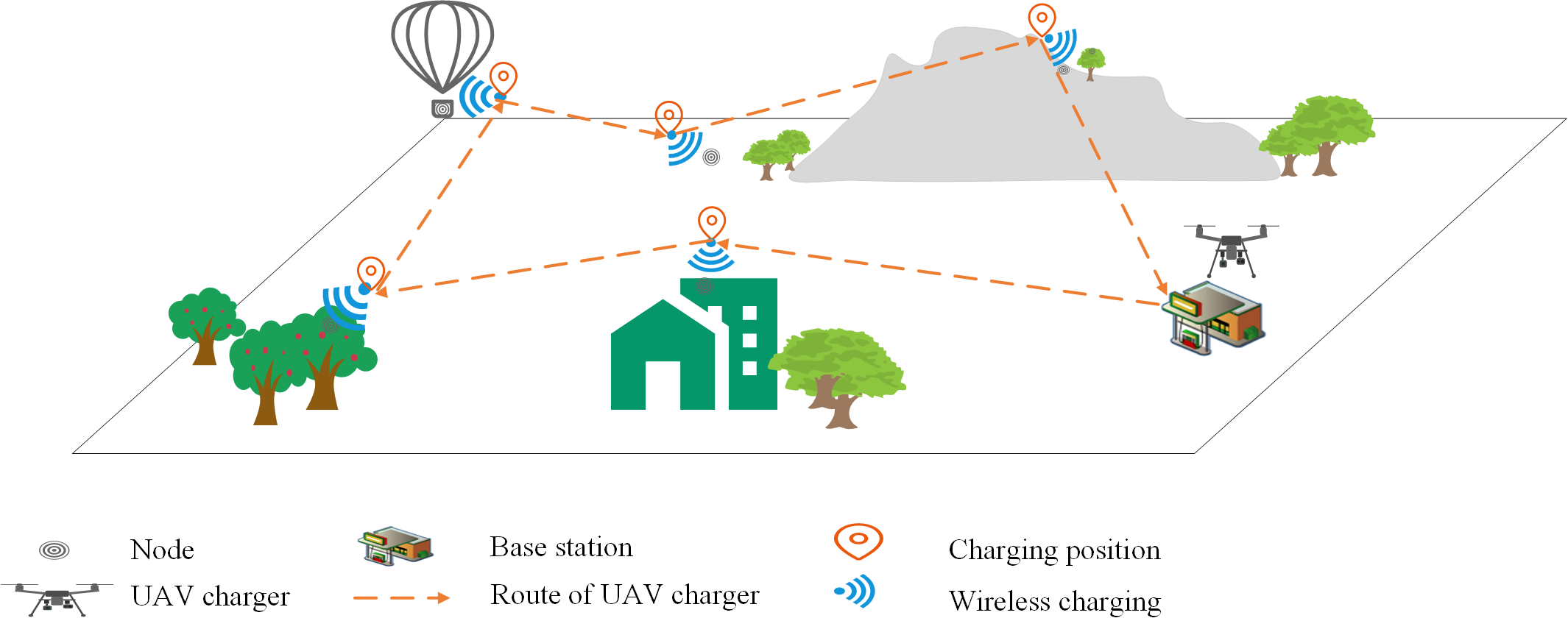}
\caption{An example 3D-WRSN.}
\label{fig_system_model}
\vspace{-0.8cm}
\end{figure}

\subsection{Energy Consumption Model of the UAV}
\label{s3_energy_cost_model_for_uav}

When fulfilling a charging schedule, the energy consumption of the UAV contains three parts: flying energy consumption, hovering energy consumption, and charging energy consumption (which results from charging the nodes). 

The energy consumption power of UAV when flying at speed $v$ is modeled as Eq.~\eqref{eq_uav_power_consumption_model}\cite{Lin2021TMC}, where $p_b$ and $p_i$ are two constants representing the blade profile power and induced power. $U_\text{tip}$ represents the tip speed of the rotor. $v_0$ is the average rotor induced speed in hovering. $d_0$ and $s_\text{UAV}$ are the fuselage drag coefficient and rotor solid ratio, respectively. $r$ and $A$ represent the air density and rotor disc area, respectively.  % need to use the original ref
\begin{equation}
\label{eq_uav_power_consumption_model}
\begin{split}
p_\text{UAV}(v){=}&p_{b}\left(1{+}\frac{3v^{2}}{U_\text{tip}^{2}}\right)
{+}p_{i}\left(\sqrt{1{+}\frac{v^{4}}{4v_{0}^{4}}}{-}\frac{v^{2}}{2 v_{0}^{2}}\right)^{1/2} \\
&{+}\frac{1}{2} d_{0}{\rho}s_{\text{UAV}}A{v^3}.
\end{split}
\end{equation}

According to Eq.~\eqref{eq_uav_power_consumption_model}, the UAV's flying energy consumption power $p_\text{Fly}$ is $p_\text{Fly}{:=}p_\text{UAV}(\bar{v})$, where $\bar{v}$ is the UAV's constant flying speed. $p_\text{Hov}{=}p_\text{UAV}(0)$ is the hovering energy consumption power.  
Let $\mathcal{S}_\text{DirRep}(i){:=}\{\psi_1,\psi_2,\ldots,\psi_{k_i}\}$ denote the charging direction set at position $l_i$, then a position-direction (Pos-Dir) pair $(l_i,\psi_j)$ completely defines a direction and its associated position. When the UAV is charging at position $l_i$ along direction $\psi_j$, we say that the UAV is charging along (or using) the Pos-Dir pair $(l_i,\psi_j)$. Let define $\mathcal{S}_\text{PosDir}(l_i){:=}\{(l_i,\psi_j)|i{\in}\{1,2,\ldots,k_i\}$ and $\mathcal{S}_\text{PosDir}{:=}$ ${\cup}_{i{=}1}^{|\mathcal{L}|}\mathcal{S}_\text{PosDir}(l_i)$. We assume that the pairs in $\mathcal{S}_\text{PosDir}$ are already properly sorted, firstly on position and then on direction value, and assume the number of charging directions at $l_0$ is $k_0{=}0$. . 

Let $\textbf{t}^\text{UAV}_\text{Chrg}{=}[t^\text{Chrg}_1,t^\text{Chrg}_2,\ldots,t^\text{Chrg}_{K}]^\text{T}$ with $K{:=}\sum_{i{=}1}^{n}k_i$ denote the charge time vector corresponding to the Pos-Dir pairs in $\mathcal{S}_\text{PosDir}$. As the UAV always transmit energy with power $p_0$, the charging energy consumption is Eq.~\eqref{eq_uav_egy_chrg}. During the charging period at a position, the UAV should always hovering there, thence the corresponding hovering energy consumption can be obtained as Eq.~\eqref{eq_uav_egy_hover}.

\begin{align}
e^{\text{UAV}}_\text{Chrg}{:=}&p_0{\cdot}{\sum}_{i{=}1}^{K}t^\text{Chrg}_i
{=}p_0\mathds{1}^{1{\times}K}\textbf{t}^\text{UAV}_\text{Chrg}.\label{eq_uav_egy_chrg}\\
e^{\text{UAV}}_\text{Hov}{:=}&p_\text{UAV}(0){\cdot}{\sum}_{i{=}1}^{K}t^\text{Chrg}_i
{=}p_\text{UAV}(0)\mathds{1}^{1{\times}K}\textbf{t}^\text{UAV}_\text{Chrg}.\label{eq_uav_egy_hover}
\end{align}

Given a charging tour $\text{r}(\mathcal{L}){:=}[l_0,l_{\pi_1},l_{\pi_2},\ldots,l_{\pi_|\mathcal{L}|},l_{\pi_{(|\mathcal{L}|{+}1)}}$ ${=}l_0]$, the flying energy consumption and flying time of the UAV can be expressed as

\begin{align}
e^{\text{UAV}}_\text{Fly}(\text{r}(\mathcal{L}))&=\sum_{i{=}0}^{|\mathcal{L}|}
p_\text{Fly}*d(l_i,l_{i{+}1})/\bar{v}*w(l_i,l_{i{+}1}),\label{eq_uav_tour_fly_egy}\\
t^{\text{UAV}}_\text{Fly}(\text{r}(\mathcal{L}))&=\sum_{i{=}0}^{|\mathcal{L}|}d(l_i,l_{i{+}1})/\bar{v}.\label{eq_uav_tour_fly_time}
\end{align}

The total energy consumption and the remaining energy of the UAV after fulfilling a charging schedule can be expressed by Eqs.~\eqref{eq_uav_egy_cost_total} and~\eqref{eq_uav_egy_final}, respectively.

\begin{align}
e^{\text{UAV}}_\text{Total}{=}&e^{\text{UAV}}_\text{Fly}(\textbf{r}(\mathcal{L})){+}e^{\text{UAV}}_\text{Chrg}{+}e^{\text{UAV}}_\text{Hov},\label{eq_uav_egy_cost_total}\\
e_\text{F0}{=}&e_\text{B0}{-}e^{\text{UAV}}_\text{Total}.\label{eq_uav_egy_final}
\end{align}

\subsection{WPT Energy Transfer Model}
\label{sec_wpt_egy_transfer_model}

Energy transfer coefficient is usually determined using energy transfer coefficient model. Here we adopt the model utilized in \cite{liang2021charging,liu2022ITJ}, where the region covered by UAV's energy transmission signal is the joint of a sphere with radius $D$ and a cone with cone apex angle $\phi$ and height larger than $D$, both rooted at the UAV, as depicted in Fig.~\ref{fig_charging_cone}. Although the bottom surface of the region is not a plane, we still call it as charging cone. The center-line of the cone is referred as the corresponding charging direction, denoted as $\psi$. With the additional parameter including node direction $\theta$, and the distance $d$ from the UAV to the node concerned, Eq.~\eqref{eq_uav_egy_tranfer_coeff_model} determines the energy transfer coefficient from the UAV to the node, where $\angle(\psi,\theta)$ represents the angle between the two directions $\psi$ and $\theta$. They both expressed as three-element vectors. 
We assume that the cone angle $\phi$ remains constant, whereas charging direction $\psi$ is freely adjustable. 

\begin{figure}
    \centering
    \includegraphics[width=0.40\textwidth]{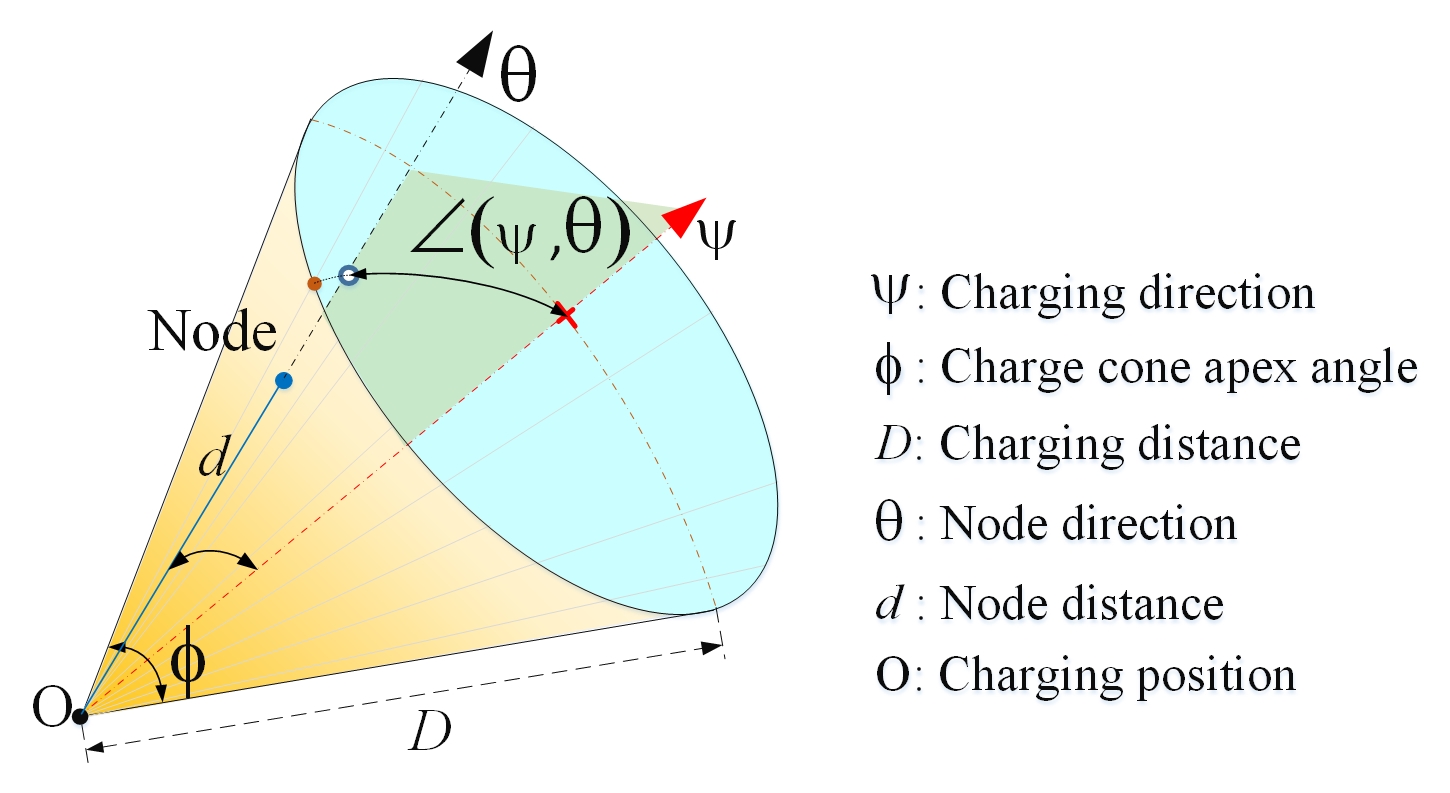}
    \caption{Charging cone diagram}
    \label{fig_charging_cone}
\end{figure}

\begin{equation}
\label{eq_uav_egy_tranfer_coeff_model}
c(\psi,\phi,\theta,d){=} 
\left\{
\begin{array}{ll}
\frac{\delta}{(\alpha{+}d)^{\beta}},&d{\le}D, \angle(\psi,\theta){\leq}\phi/2;\\
0,&\text{otherwise}. \\
\end{array}
\right.
\end{equation}

Table \ref{symbolDefine} lists main symbols for facilitating the readers.
\begin{table}[tbp]
\caption{Symbol Definitions}
\begin{center}
    \small
    \begin{tabular}{|p{0.07\textwidth}|p{0.365\textwidth}|}
        \hline \hline
        \textbf{Symbol} & \textbf{Definition} \\
        \hline
        \hline
        $\mathcal{U}$,$n$& $\mathcal{U}{:=}\{u_1,u_2,\ldots,u_n\}$ is the set of $n$ nodes \\
        \hline
        $\mathcal{L}$ & $\mathcal{L}{:=}\{l_1,l_2,\ldots,l_n\}$ is the set of nodes' positions, with $l_i{:=}(x_i,y_i,z_i)$ is node $u_i$'s position\\
        \hline
        $\textbf{e}_\text{B}$,$\textbf{e}_\text{U}$,$\textbf{e}_\text{D}$,
        $\textbf{e}_\text{R}$, $\textbf{e}_\text{F}$,  & The vectors of the nodes' initial energy, battery capacity (upper limit), energy demand, energy received from WPT, final energy  \\
        \hline
        $e_\text{B0}$,$e_\text{U0}$, $\bar{v}$, $l_0$ & The UAV's initial energy, battery capacity, flying speed, initial position (i.e., BS's position)\\
        \hline
        $\psi$,D,$\phi$ & The UAV's charging direction, charging distance, charge cone apex angle (cone angle)\\
        \hline
        $p_0$,$p_\text{Fly}$,
        $p_\text{Hov}$ & The UAV's constant energy transmission power, hovering energy consumption power, flying energy consumption power\\
        \hline
        $\mathcal{N}_\text{ChrgPos}$ & Set of charging positions  \\
        \hline
        $\mathcal{N}_\text{Sphere}$ $(O)$ & Set of nodes enclosed in the sphere centered at $O$ with radius $D$  \\
        \hline
        $\mathcal{N}_\text{Cone}$ $(O,\textbf{v},\phi)$ &Set of nodes covered by a charging cone rooted at $O$ with direction \textbf{v}, cone angle $\phi$, and height $D$ \\   
        \hline
        $\mathcal{S}_\text{DirRep}$ $(O,A)$ & Set of selected representative directions at position $O$ with reference node $A$ in context $\mathcal{C}_\text{ChrgPos}(O,A)$ \\        \hline
        $\mathcal{S}_\text{DirRep}(O)$ & Set of selected representative directions at position $O$ in context $\mathcal{C}_\text{ChrgPos}(O)$ \\
        \hline
        $\mathcal{S}_\text{PosDir}$& Set of Pos-Dir pairs defined in Eq.~\eqref{Eq_construct_dirrep_all} as three-element tuples, including charging position, charging direction, and its covered node set\\
        \hline
        $\mathcal{C}_\text{ChrgPos}$ $(O,A)$ & Context with charging position $O$ and reference node $A$ \\
        \hline
        $\mathcal{C}_\text{ChrgPos}(O)$ & Context with charging position $O$ \\
        \hline
        $\mathcal{S}_\text{LMaxRng}$ $(O,A)$ & Set of LMax-SCN angle ranges in context $\mathcal{C}_\text{ChrgPos}(O,A)$ \\
        \hline
        $\textbf{s}^\text{CTS}$, $\tau^\text{CTS}$ & Charging schedule and its time span  \\
        \hline
        $e^\text{UAV}_\text{Fly}(\textbf{s})$, $e^\text{UAV}_\text{Hov}(\textbf{s})$, $e^\text{UAV}_\text{Chrg}(\textbf{s})$, $e^\text{Total}_\text{Loss}(\textbf{s})$ & Flying energy consumption, hovering energy consumption, and charging energy consumption for conducting a charging schedule $\textbf{s}$, $e^\text{Total}_\text{Loss}(\textbf{s}){:=}e^\text{UAV}_\text{Fly}(\textbf{s}){+}e^\text{UAV}_\text{Hov}(\textbf{s}){+}e^\text{UAV}_\text{Chrg}(\textbf{s}){-}e^\text{Nodes}_\text{Rcv}(\textbf{s})$  \\
        \hline
        $e^\text{UAV}_\text{Total}$
        $e^\text{UAV}_\text{Hov}(\textbf{t})$, $e^\text{UAV}_\text{Chrg}(\textbf{t})$, $e^\text{Nodes}_\text{Rcv}
        (\textbf{t})$ & Total energy consumption ,hovering energy consumption and charging energy consumption of UAV, and energy actually received by the nodes, which are depend on charging time vector $\textbf{t}$ along Pos-Dir pairs in $\mathcal{S}_\text{PosDir}$  \\
        \hline
        $\mathcal{S}^\text{CTS}_\text{Fly}$, $\mathcal{S}^\text{CTS}_\text{Chrg}$,& Flying schedule item set and charging schedule item set embedded in a charging schedule\\
        \hline
        $\mathbf{C}_\text{ETC}$ & $\mathbf{C}_\text{ETC}{:=}[c(i,j)]_{\{i{\in}\{1,2,\cdots,|\mathcal{S}_\text{PosDir}|\},U_j{\in}\mathcal{U}}$\\ 
        \hline
        $\textbf{t}^\text{UAV}_\text{Chrg}$ & UAV's charging time vector along the Pos-Dir pairs in $\mathcal{S}_\text{PosDir}$ \\
        \hline\hline
    \end{tabular}
\end{center}
\label{symbolDefine}
\vspace{-0.8cm}
\end{table}

\section{The DCS-3D Problem}
\label{sec_problem_formulation}

\subsection{Structure of Charging Schedule}
\label{s4_structure_of_charging_shcedule}
A complete charging schedule solving a DCS-3D problem instance should determine the UAV's tour schedule and charging time schedule along the Pos-Dir pairs. We describe a complete schedule using an ordered list called Charging Tour Schedule (CTS) list. 

We use $\textbf{s}^\text{CTS}{:=}[s^\text{C}_1, s^\text{C}_2, \ldots, s^\text{C}_{m}]$ to represent a CTS list with $m$ items. The $i$-th item $s^\text{C}_i{:=}(state,x,t,\textbf{v})$ has three fields, where $state$ is a binary variable. The $state$ value indicates the type of the item: $state{=}1$ indicates it as a flying schedule item for arranging the travel between charging positions, whereas $state{=}0$ means that it is a charging schedule item for arranging the energy transmission operation along Pos-Dir pairs. For schedule item $(state,x,t,\textbf{v})$, if $state{=}1$ then it means that the UAV should fly toward $x$ for $t$ time, and if $state{=}0$ then it means that the UAV should charging along direction $\textbf{v}$ for time $t$. We define two sets as $\mathcal{S}^\text{CTS}_\text{Fly}{=}\{s^\text{C}_i|s^\text{C}_i.state{=}0\}$ and $\mathcal{S}^\text{CTS}_\text{Chrg}{=}\{s^\text{C}_i|s^\text{C}_i.state{=}1\}$. 

When fulfilling a charging schedule $\mathbf{s}^\text{CTS}$, its items are executed sequentially. Let $\tau_{i}^\text{C}{:=}\sum_{j{=}1}^{i}s_j^\text{C}.t$, then item $s_i^\text{C}$ starts at time $\tau_{i{-}1}^\text{C}$ and ends at time $\tau_{i}^\text{C}$. The time span of the schedule $\textbf{s}^\text{CTS}$ is $\tau^\text{CTS}{:=}\sum_{j{=}1}^{|\textbf{s}^\text{CTS}|}s_j^\text{C}.t$.
\vspace{-0.5cm}

\subsection{The DCS-3D Problem}
\label{s4_problem_statement}

For an instance of the DCS-3D problem, the total initial energy $e_\text{TB}{:=}e_\text{B0}{+}\mathds{1}^{1{\times}n}\textbf{e}_\text{B}$ is a constant. 
After fulfilling a charging schedule $\textbf{s}$, the total remaining energy of all nodes and the UAV is $e_\text{TF}(\textbf{s}){:=}e_\text{F0}{+}\mathds{1}^{1{\times}n}\textbf{e}_\text{F}$. For prevously defined energy consumption symbols, such as $e^\text{UAV}_\text{Hov}$, $e^\text{UAV}_\text{Chrg}$, $e^\text{Nodes}_\text{Rcv}$, we can re-write them by appending $\textbf{s}$ to emphasize the dependence on $\textbf{s}$. The total energy loss incurred during the execution of schedule $\textbf{s}$ is thus computed as:
\begin{equation}
\label{eq_egy_loss_total}
e^\text{Total}_\text{Loss}(\textbf{s}){=}e_\text{TB}{-}e_\text{TF}(\textbf{s}){=}e_\text{Total}^\text{UAV}(\textbf{s}){-}e_\text{Rcv}^\text{Nodes}(\textbf{s})
\end{equation}

As $e^\text{UAV}_\text{Hov}(\textbf{s})$, $e^\text{UAV}_\text{Chrg}(\textbf{s})$ and $e^\text{Nodes}_\text{Rcv}(\textbf{s})$ are all completely determined by the charging time list $\textbf{t}$ embedded in schedule $\textbf{s}$, they can be expressed as $e^\text{UAV}_\text{Hov}(\textbf{t})$ $e^\text{UAV}_\text{Chrg}(\textbf{t})$, and $e^\text{Nodes}_\text{Rcv}(\textbf{t})$, respectively. Let define $e^\text{WPT+Hov}_\text{Loss}(\textbf{t})$ as 
\begin{equation}
\label{eq_egy_loss_wpt_hov}
e^\text{WPT+Hov}_\text{Loss}(\textbf{t}){:=}e^\text{UAV}_\text{Hov}(\textbf{t}){+}e^\text{UAV}_\text{Chrg}(\textbf{t}){-}e^\text{Nodes}_\text{Rcv}(\textbf{t})
\end{equation}

The objective interested here is to minimize the total energy loss with the prerequisite of satisfying the energy demands of all nodes. With this objective, the targeted DCS-3D Problem can be formally stated as follows. 

\textbf{DCS-3D Problem: Given an 3D-WRSN contains a UAV, $n$ nodes at positions in $\mathcal{L}$, a BS at $l_0$, and together with the preliminary models, the task is to find a charging schedule $\mathbf{s}^\text{CTS}$ leading to minimum energy loss $\mathbf{e}^\text{Total}_\text{Loss}$ while guaranteeing that all charging demands are satisfied.}

DCS-3D can be formulated as~Eq.\eqref{eq_DWRACS_problem_math}, where $\Omega_\text{r}(\mathcal{L})$ denotes the solution space of valid charge tours traversing all positions in $\mathcal{L}$, $\mathcal{\psi}$ is the space of charging directions, and ~Eq.\eqref{eq_DWRACS_problem_math_C1} assures that all nodes' energy demands are satisfied.

\begin{subequations}
\label{eq_DWRACS_problem_math}
\begin{align}
\textbf{(P1)} \quad 
& \mathop{\min}\limits_{
\tiny
\makecell{
\textbf{r}(\mathcal{L}),\,
\mathcal{S}_\psi(\mathcal{L}),\,
\textbf{t}^\text{UAV}_\text{Chrg}
}
}
&& e_\text{Loss}^\text{Total}
\!\left(
\textbf{r}(\mathcal{L}),
\mathcal{S}_\psi(\mathcal{L}),
\textbf{t}^\text{UAV}_\text{Chrg}
\right), \\[2mm]
& \text{s.t.} 
&& \textbf{e}_\textbf{F} {\geq} \textbf{e}_\textbf{B} {+} \textbf{e}_\textbf{D}, 
\label{eq_DWRACS_problem_math_C1} \\[1mm]
&&& \textbf{t}^\text{UAV}_\text{Chrg} {\geq} \textbf{0},
\label{eq_DWRACS_problem_math_C2} \\[1mm]
&&& \textbf{r}(\mathcal{L}) {\in} \Omega_\mathcal{L}, \quad 
\mathcal{S}_\psi(\mathcal{L}) {\in} \Omega_\psi.
\label{eq_DWRACS_problem_math_C3}
\end{align}
\end{subequations}

However, the following theorem tells that solving the DCS-3D problem is challenging.
 
\begin{theorem}
The DCS-3D problem is NP-hard.
\end{theorem}
\begin{proof}
Proof is provided in Sec.I in the Appendix. 
\end{proof}
Given the NP-hard nature of DCS-3D, we focus on developing approximate solutions to address its two core challenges: (1) infinite charging direction space and (2) UAV path planning . Building upon \cite{Gao2024TMC}, we resolve the first challenge by deriving a minimal FuncEqv direction set through our cMFEDS algorithm (Section~\ref{sec_address_direction_challenge}). For the UAV path planning challenge, we adapt the LKH heuristic \cite{lkh_2019} to develop the FELKH-3D solution, detailed in Section~\ref{sec_algorithm}.

\section{Addressing the Infinite Charging Direction Space Challenge}
\label{sec_address_direction_challenge}

\subsection{Preliminaries}
For a vector $\textbf{v}$, let $|\textbf{v}|$ be its norm and  $\text{normalize}(\textbf{v}){:=}\textbf{v}/|\textbf{v}|$. For any two nodes $A$ and $B$, let $\overrightarrow{AB}$ represent the vector from $A$ to $B$, let $|AB|$ be the length of vector $\overrightarrow{AB}$, and let $\textbf{e}_{AB}$ denote the corresponding unit vector. Thus $\overrightarrow{AB}{=}|AB|\textbf{e}_{AB}$. In later equations, For any node, say $A$, its symbol $A$ is also used to refer to the vector of its 3D coordinates. Given a charging position $O$ and two nodes $A$ and $B$, let $\varphi_{AOB}$ denote the radian angle value of $\angle{AOB}$. As the nodes have fixed positions, \textit {node A} and \textit{position A} are used interchangeably. For a node $A$, let $A(x)$, $A(y)$ and $A(z)$ denote its x, y and z coordinates, respectively.   

Let $\mathcal{N}_\text{Sphere}(O)$ denote the set of the nodes contained in the sphere centered at position $O$ with radius $D$ (the charging distance), and let $\mathcal{N}_\text{Cone}(O,\textbf{v},\phi)$ denote the set of the nodes covered by a charging cone rooted at $O$ with charging direction $\textbf{v}$ and cone angle $\phi$. This node set is called as \textbf{the node set covered by direction $\textbf{v}$} for short. All nodes in $\mathcal{N}_\text{Cone}(O,\textbf{v},\phi)$ are said as covered by direction $\textbf{v}$. A node on the boundary of a charging cone is also regarded as covered by the corresponding charging direction. 

    In the following text, discussions are usually made in the context with a certain charging position $O$. We denote a such context as $\mathcal{C}_\text{ChrgPos}(O)$, which is assumed as the default context in later text. Furthermore, when just considering charging cones with a certain node, say $A$, always on the boundary, we denote the context as $\mathcal{C}_\text{ChrgPos}(O,A)$. Such a node $A$ in context $\mathcal{C}_\text{ChrgPos}(O,A)$ is called the reference node in the context. 

In context $\mathcal{C}_\text{ChrgPos}(O)$ or even $\mathcal{C}_\text{ChrgPos}(O,A)$, for a node set $s_1$, if there is no charging direction $\textbf{v}_1$ such that $s_1{\subset}\mathcal{N}_\text{Cone}(O,\textbf{v}_1,\phi)$, we say that $s_1$ is a \textbf{Local Maximum Simultaneous Charging Node set (LMax-SCN set)}, otherwise it is a \textbf{non-LMax-SCN} set. A charging direction covering an LMax-SCN set is called an \textbf{LMax-SCN} direction, otherwise a non-LMax-SCN direction. 

    As in Ref.~\cite{Gao2024TMC}, in context $\mathcal{C}_\text{ChrgPos}(O)$, we classify the relationship between charging directions according to the partial order relationship between their corresponding covered node set. To be specific, for two directions $\textbf{v}_1$ and $\textbf{v}_2$, if $\mathcal{N}_\text{Cone}(O,\textbf{v}_1,\phi){\supset}$ $\mathcal{N}_\text{Cone}(O,\textbf{v}_2,\phi)$, we say that $\textbf{v}_1$ is functional better than $\textbf{v}_2$, or $\textbf{v}_1$ outperforms $\textbf{v}_2$, and denote this relationship as $\textbf{v}_1{\rhd}\textbf{v}_2$, or $\textbf{v}_2{\lhd}\textbf{v}_1$. If $\mathcal{N}_\text{Cone}(O,\textbf{v}_1,\phi){=}\mathcal{N}_\text{Cone}(O,\textbf{v}_2,\phi)$, then we say that they are Functional Equivalent (FuncEqv), and denote the situation as $\textbf{v}_1{\equiv}\textbf{v}_2$. We use $\textbf{v}_1{\unlhd}\textbf{v}_2$ to mean either $\textbf{v}_1{\equiv}\textbf{v}_2$ or $\textbf{v}_1{\lhd}\textbf{v}_2$. If either of the relations of $\textbf{v}_1{\lhd}\textbf{v}_2$, $\textbf{v}_1{\equiv}\textbf{v}_2$, and $\textbf{v}_1{\rhd}\textbf{v}_2$ is true, we state that $\textbf{v}_1$ and $\textbf{v}_2$ are comparable, otherwise non-comparable.

Given two direction sets $S_{V1}$ and $S_{V2}$, if for each direction $\textbf{v}_1{\in}S_{V1}$, there is a $\textbf{v}_2{\in}S_{V2}$ such that $\textbf{v}_1{\unlhd}\textbf{v}_2$, then we denote the relationship as $S_{V1}{\unlhd}S_{V2}$. We use $S_{V1}{\equiv}S_{V2}$ to mean that both $S_{V1}{\unlhd}S_{V2}$ and $S_{V1}{\unrhd}S_{V2}$ holds. If $S_{V1}{\unlhd}S_{V2}$ and ${\exists}\textbf{v}_2{\in}S_{V2}$ such that, for any $\textbf{v}_1{\in}S_{V1}$ that is comparable with $\textbf{v}_2$, there is $\textbf{v}_1{\lhd}\textbf{v}_2$, then we say that $S_{V2}$ is functional better than $S_{V1}$, and denote the relationship as $S_{V1}{\lhd}S_{V2}$. 

Solving charging tasks involves charging the nodes and moving between the charging positions. Even if the node set covered by a charging direction at a charging position is a subset of one direction at another charging position, the two Pos-Dir pairs may both be helpful in serving the tasks. Hence we emphasize that the directions at different charging positions are regarded as not comparable.

We can easily obtain the following results using above definitions, and the proofs are omitted for their trivialness.
\begin{lemma}
\label{lemma_equiv_delete_worse}
Given a direction set $S_{V1}$ and a direction $\textbf{v}_1{\in}S_{V1}$, let $S_\text{Worse}{:=}\{\textbf{v}|\textbf{v}{\in}S_{V1}, \textbf{v}{\lhd}\textbf{v}_1\}$ and $S_{V2}{:=}S_{V1}{-}S_\text{Worse}$, then $S_{V1}{\equiv}S_{V2}$.
\end{lemma}

\begin{lemma}
\label{lemma_equiv_select_representative}
Given a direction set $S_{V1}$ and a direction $\textbf{v}_1{\in}S_{V1}$, let $S_\text{Equiv}{:=}\{\textbf{v}|\textbf{v}{\in}S_{V1}, \textbf{v}{\equiv}\textbf{v}_1\}$ and $S_{V2}{:=}S_{V1}{-}S_\text{Equiv}{+}\{\textbf{v}_1\}$, then $S_{V1}{\equiv}S_{V2}$.
\end{lemma}

In context $\mathcal{C}_\text{ChrgPos}(O)$, the valid charging directions fill a continuous sphere surface, and can be denoted as $\mathcal{S}_\text{Dir,Sphere}{:=}\{\textbf{v}|v{\in}\mathbb{R}^3,|\textbf{v}|=1\}$. Inspired by Ref.~\cite{Gao2024TMC}, to address the infinite charging direction space challenge at a charging position, we tried to find a minimum discrete direction set $\mathcal{S}_\text{DirRep}$ that is FuncEqv with the set $\mathcal{S}_\text{Dir,Sphere}$. This task can be formulated as Eq.~\eqref{Eq_problem_find_min_funceqv}. 
\vspace{-0.2cm}
\begin{subequations}
\label{Eq_problem_find_min_funceqv}
\begin{align}
(\textbf{P2})&\mathop{\min}\limits_{\mathcal{S}_\text{DirRep}}&|\mathcal{S}_\text{DirRep}|,\hspace{24mm}\\
&\text{s.t.}&{\mathcal{S}_\text{DirRep}{\equiv}\mathcal{S}_\text{Dir,Sphere}},\hspace{10mm}\\
& &\mathcal{S}_\text{DirRep}{\subseteq}\mathcal{S}_\text{Dir,Sphere}.\hspace{10mm}
\end{align}
\end{subequations}

In the following text, we try to solve \textbf{P2} by firstly analyzing how to find LMaxSCN directions in context $\mathcal{C}_\text{ChrgPos}(O,A)$, and then, in the extended context $\mathcal{C}_\text{ChrgPos}(O)$, to efficiently obtain $\mathcal{S}_\text{DirRep}$ by combining and refining the charging direction sets using different reference nodes. The $\mathcal{S}_\text{DirRep}$ sets at different charging positions can be combined straightforwardly into a final charging direction set for charging scheduling. 

In context $\mathcal{C}_\text{ChrgPos}(O)$, we assume that all nodes in $\mathcal{N}_\text{Sphere}(O)$ are already projected onto a unit sphere centered at $O$ using Eq.~\eqref{Eq_node_proj}. So in the following text, all nodes involved are assumed as the nodes after projection by default. 

In particular, for a single node $A{\in}\mathcal{N}_\text{Sphere}(O)$, the valid charging direction can be uniquely determined by the vector $\overrightarrow{OA}$, which directly connects the charging position $O$ and the node $A$. This direction is thus regarded as the intrinsic charging direction of node $A$. However, since this one-to-one correspondence yields only trivial single-node coverage, the following analysis focuses on nontrivial cases where a single charging direction simultaneously covers multiple nodes. This enables a more efficient representation and reduction of the continuous directional space.

\begin{equation}
\label{Eq_node_proj}
A{=}O{+}\text{normalize}(\overrightarrow{OA}),\quad{\forall}A{\in}\mathcal{N}_\text{Sphere}(O)
\end{equation}

\subsection{Determine the Projected Angle Range Covering a Node}

\begin{lemma}
\label{lemma_end_point}
In context $\mathcal{C}_\text{ChrgPos}(O)$, for any two nodes $A,B{\in}\mathcal{N}_\text{Sphere}(O)$ with ${\varphi_{AOB}}{\leq}\phi$, the charging direction $\textbf{e}_\text{O,Bound}(A,B)$, with both $A$ and $B$ are exactly on the boundary of the corresponding charging cone, can be expressed as $\textbf{e}_\text{O,Bound}(A,B)$ ${=}\text{normalize}(\textbf{v})$, where vector $\textbf{v}$ is given in 

\begin{eqnarray}
\label{Eq_vec_dir}
\begin{aligned}
\textbf{v}{=}&\sec(\varphi_{AOB}/2){\cdot}\text{normalize}(\textbf{e}_{OA}{+}\textbf{e}_{OB})\\
&{\pm}\sqrt{\tan^2(\phi/2){-}\tan^2(\varphi_{AOB}/2)}{\cdot}\text{normalize}(\textbf{e}_{OB}{\times}\textbf{e}_{OA}).
\end{aligned}
\end{eqnarray}

\end{lemma}

\begin{proof}
This can be proved easily using geometric knowledge as depicted in Fig.~\ref{fig_end_direction}. Detailed proof is provided in Sec.II in the Appendix.
\end{proof}

\begin{figure}[!htbp]
\centering
\includegraphics[width=0.6\linewidth]{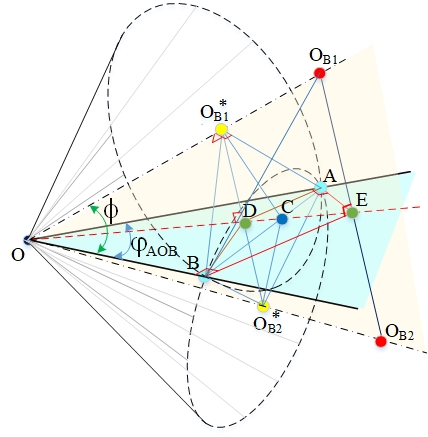}
\caption{Sketch for calculating the two critical directions covering two particular nodes \textit{A} and \textit{B}}
\label{fig_end_direction}
\vspace{-0.3cm}
\end{figure}
In context $\mathcal{C}_\text{ChrgPos}(O,A)$, where node $A$ is the \textit{reference node}, only the charging cones with node node $A$ on its boundary are considered. We create a plane perpendicular to $\overrightarrow{OA}$ while contains $A$, and denote it as $\beta$. We project the charging directions into plane $\beta$ and analyze their relationships in it. For a charging direction $\textbf{e}$, we denote the intersection point of the vector line $\textbf{v}$ (originating from $O$) and plane $\beta$ as point $O_v$, and call it the \textit{projection point} of $\textbf{e}$, and call the vector $\overrightarrow{AO_v}$ as $\textbf{v}$'s \textit{projected direction}. 

To facilitate describing the projected directions, we will setup a reference direction in plane $\beta$ and associate the projected directions with the corresponding angle values relative to the reference direction, with positive in counter-clockwise direction facing $\overrightarrow{OA}$. We denote the reference direction as $\textbf{e}_\text{REF}$. In this paper, we set $\textbf{e}_\text{REF}$ as the intersection line of plane $\beta$ with the horizontal plane with $z{=}A(z)$ towards right. It is easy to note that $\textbf{e}_\text{REF}$ must be perpendicular with both $\overrightarrow{OA}$ and the $z$ axis. So, additionally with the right hand rule for determining the direction of vector $\times$ operation, there is $\textbf{e}_\text{REF}{=}\overrightarrow{OA}{\times}[0,0,1]$. When these two planes are the same, $\textbf{e}_\text{REF}$ can be selected randomly in the plane, and this case is omitted for it is trivial. 

With a fixed $\textbf{e}_\text{REF}$, a direction in plane $\beta$ is associated with an angle in range $[0,2\pi)$. Thus, a charging direction $\textit{e}_{OX}$ is associated with its \textit{projection point} $O_X$, \textit{projected direction} $\overrightarrow{AO_X}$, and its \textit{projected angle} $\theta_{AO_X}$. In context $\mathcal{C}_\text{ChrgPos}(O,A)$, as there is a one-to-one map between a charging direction $\textbf{v}$ and its projected angle $\theta$, we treat $\mathcal{N}_\text{Cone}(O,\theta,\phi)$ and $\mathcal{N}_\text{Cone}(O,\textbf{v},\phi)$ as the same thing.

\begin{lemma}
\label{lemma_angle_range}
In context $\mathcal{C}_\text{ChrgPos}(O,A)$, given a reference direction $\textbf{e}_\text{REF}$ in plane $\beta$, a charging direction $\textbf{e}_{OX}$'s projection point $O_X$, projected direction $\overrightarrow{AO_X}$, and angle $\theta_{AO_X}$ are provided by the equations from Eq.~\eqref{Eq_proj_point} to Eq.~\eqref{Eq_proj_angle}, respectively. Given an angle $\theta_X$, the corresponding projected direction $\textbf{e}_{\theta_X}$ and projection point $O_{\theta_X}$ are given by Eq.~\eqref{Eq_angle2_dir} and Eq.~\eqref{Eq_angle2_point}, respectively.
\vspace{-0.5cm}
\begin{align}
O_X{=}&O{-}\frac{1}{\textbf{e}_{OX}{\cdot}\overrightarrow{OA}}\textbf{e}_{OX},\label{Eq_proj_point}\\
\overrightarrow{AO_X}{=}&\overrightarrow{AO}{-}\frac{1}{\textbf{e}_{OX}{\cdot}\overrightarrow{OA}}\textbf{e}_{X},\label{Eq_proj_dir}\\
|AO_X|{=}&\tan(\phi),\label{Eq_proj_dir_len}
\end{align}

\begin{align}
\theta_{AO_X}{=}&
\begin{cases}
\arccos(\frac{\overrightarrow{AO_X}{\cdot}\textbf{e}_\text{REF}}{|\overrightarrow{AO_X}||\textbf{e}_\text{REF}|}), & \text{if}\hspace{2mm}{AO_X}(z){\geq}0;\\
\arccos(\frac{\overrightarrow{AO_X}{\cdot}\textbf{e}_\text{REF}}{|\overrightarrow{AO_X}||\textbf{e}_\text{REF}|}){+}\pi, & \text{if}\hspace{2mm}{AO_X}(z){<}0.\\
\end{cases}\label{Eq_proj_angle}\\
\textbf{e}_{\theta_X}{=}&\cos(\phi){\cdot}\textbf{e}_\text{REF}{+}\cos(\theta{-}\pi/2){\cdot}\text{normalize}(\overrightarrow{OA}{\times}\textbf{e}_\text{REF}),\label{Eq_angle2_dir}\\
O_{\theta_X}{=}&A{+}\textbf{e}_{\theta_X}\tan{\phi}.\label{Eq_angle2_point}
\end{align}

\end{lemma}

\begin{proof}
This can be proved easily with the help of Fig.~\ref{fig_ref_range}.  Detailed proof is provided in Sec.III in the Appendix.
\end{proof}

\begin{figure}[!htbp]
\centering
\includegraphics[width=0.6\linewidth]{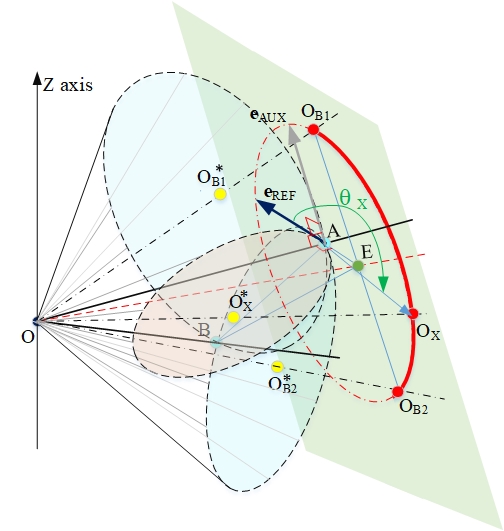}
\caption{Sketch of the range of the projected angles for covering node $B$ when node $A$ is the reference node}
\label{fig_ref_range}
\vspace{-0.6cm}
\end{figure}

\subsection{Determine Local-Maximum Direction Angle Ranges}

In context $\mathcal{C}_\text{ChrgPos}(O,A)$, $A$ is the reference node. If $\mathcal{N}_\text{Cone}(O,\overrightarrow{OA},2\phi){=}\{A\}$, it implies that there are no cone directions that can simultaneously cover $A$ with any other node. In this trivial case, to assure covering node $A$, we select $\textbf{e}_\text{OA}$ into the final $\mathcal{S}_{\text{DirRep}}$.

Now let focus on the the case  $\{A\} \subset \mathcal{N}_{\text{Cone}}(O,\overrightarrow{OA},2\phi)$. For any node $B{\in}\mathcal{N}_\text{Cone}(O,\overrightarrow{OA},2\phi)$, let denote the two directions obtained in Lemma~\ref{lemma_end_point} as $\textbf{e}_\text{OB1}$ and $\textbf{e}_{\text{OB2}}$, then we can obtain their projection points $O_\text{B1}$ and $O_\text{B2}$, projected angles $\theta_\text{B1}$ and $\theta_\text{B2}$ using Lemma~\ref{lemma_angle_range}. Then, as shown in Fig.~\ref{fig_ref_range}, node $B$ is covered by any charging direction with projected angle in range $[\theta_\text{B1}$, $\theta_\text{B2}]$, i.e., the charging direction's projection point lies in the arc between $O_\text{B1}$ and $O_\text{B2}$ on the circle centered at $A$ with radius $\csc{\phi}$. We call this angle range as B's charging direction projected angle range, and call $\theta_\text{B1}$ and $\theta_\text{B2}$ as its \textit{start angle} and the \textit{end angle}, respectively. Following the same way, we generate the CPD  angle ranges for all nodes in $\mathcal{N}_\text{Cone}(O,\overrightarrow{OA},2\phi)$. Fig.~\ref{fig_ranges} provides an illustration showing multiple angle ranges, where for distinguishing overlapped angle ranges, the ranges are demonstrated on multiple circles with different radius. 

\begin{figure}[!htbp]
\centering
\includegraphics[width=0.55\linewidth]{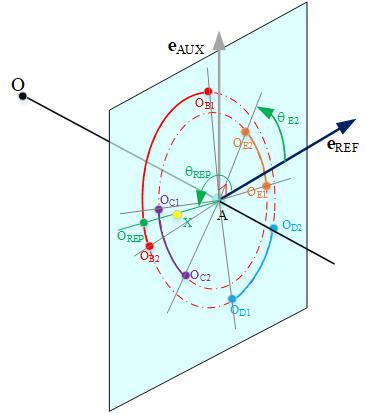}
\caption{Sketch map of multiple projected angle ranges}
\label{fig_ranges}
\end{figure}

Given a node $X$, two charging directions can be determined according to Lemma~\ref{lemma_end_point}, from which the corresponding angular range $[\theta_{X1}, \theta_{X2}]$ is derived. Although $\theta_{X1}$ is typically smaller than $\theta_{X2}$, there are cases where $\theta_{X1} {\geq} \theta_{X2}$, indicating that the angular interval crosses the reference vector $\mathbf{e}_{\text{REF}}$. 

\begin{comment}
    To address this discontinuity, the angle range is adjusted according to Eq.~\eqref{Eq_phi_modify}.

$\mathbf{e}_{\text{REF}}$.
\begin{equation}
\label{Eq_phi_modify}
\theta_{X2}{=}
\begin{cases}
\theta_{X2}{+}2\pi, & \text{if}\hspace{2mm}\theta_{X2}{<}\theta_{X1},\\
\theta_{X2},      & \text{if}\hspace{2mm}\theta_{X2}{\geq}\theta_{X1},
\end{cases}\\
,\hspace{2mm}X{\in}\mathcal{N}_\text{Cone}(O,\overrightarrow{OA},2\phi)\}.
\end{equation}
\end{comment}

Now we exploit the angle ranges for the nodes in $\mathcal{N}_\text{Cone}(O,\overrightarrow{OA},2\phi)$ for facilitating the determination of LMax-SCN directions. For each projected angle range $[\theta_\text{B1},\theta_\text{B2}]$,
$B{\in}\mathcal{N}_\text{Cone}(O,\overrightarrow{OA},2\phi)$, we represent each of the two endpoints with a tuple having structure as

\begin{equation}
\label{Eq_endpoint_tuple}
l_\text{Bi}{:=}(\theta,\delta,\eta,\sigma,\tau), \hspace{5mm}i{\in}\{1,2\},
\end{equation}
where $\theta$ stores the angle value, $\delta{\in}\{0,1\}$ is a binary variable for indicating whether $\delta$ is a start angle or an end angle. $\delta{=}0$ means a start angle, otherwise it is an end angle. $\eta$ stores the node, which is $B$ here. $\sigma{=}\theta_\text{B2}{-}\theta_\text{B1}$ records the range size. $\tau{=}[\theta_\text{B1},\theta_\text{B2}]$ stores the original angle range. These information are used in refining charging directions.

We collect and sort all the tuples in ascending order of the $\theta$ value into a list $L_\text{AngRng}{=}[l_1,l_2,\ldots,l_\kappa]$ with $\kappa{=}2*|\mathcal{N}_\text{Cone}(O,\overrightarrow{OA},2\phi)|$. Let assume that the fields of the tuples can be accessed using point operation, e.g., $l_i.\theta$ means the $l_i$'s $\theta$ value.

In the context of $\mathcal{C}_\text{ChrgPos}(O,A)$, by exploiting the projected angles of all feasible charging directions, we can construct the set $\mathcal{S}_\text{LMaxRng}(O,A)$ using Eq.~\eqref{Eq_lmax_angle_range}, which stores the angular ranges associated with potential LMax-SCN sets. Since the projected angles lie on a circular domain $[0,2\pi)$, the ordering of the tuples in $\mathcal{L}_\text{AngRng}$ induces not only consecutive start–end pairs within the list but also a wrap-around pair between the last tuple $l_\kappa$ and the first tuple $l_1$. This ensures that angle intervals spanning across the boundary between $2\pi$ and $0$ are also properly represented.

\begin{equation}
\label{Eq_lmax_angle_range}
\begin{aligned}
\mathcal{S}_\text{LMaxRng}(O,A)
= &\;\{(l_i, l_{i+1}) \mid l_i.\delta = 0,\; l_{i+1}.\delta = 1,\; 1 \le i < \kappa \} \\
&\;\cup\; \{(l_\kappa, l_1) \mid l_\kappa.\delta = 0,\; l_1.\delta = 1\}.
\end{aligned}
\end{equation}

Lemma~\ref{lemma_lmax_angle_range} and Lemma~\ref{lemma_lmax_angle_viseversa} show that, in the context of $\mathcal{C}_\text{ChrgPos}(O,A)$, each angle range in $\mathcal{S}_\text{LMaxRng}(O,A)$ corresponds to an LMax-SCN set, which can be obtained by checking the projected angle ranges of nodes in $\mathcal{N}_\text{ChrgPos}(O, O_A, \phi)$. Such an interval is referred to as an \emph{LMax-SCN angle range}, and $l_a.\eta$ and $l_b.\eta$ are termed the critical nodes associated with the angle range $[l_a.\theta,\, l_b.\theta]$.

\begin{lemma}
\label{lemma_lmax_angle_range}
In context $\mathcal{C}_\text{ChrgPos}(O,A)$, for any two projected angles $\theta_1,\theta_2{\in}[l_a.\theta,l_b.\theta]$ with $(l_a,l_b){\in}\mathcal{S}_\text{LMaxRng}(O,A)$, there is $\mathcal{N}_\text{Cone}(O,\theta_1,\phi)$ ${=}\mathcal{N}_\text{Cone}(O,\theta_2,\phi)$.
\end{lemma}

\begin{proof}
Proof is provided in Sec.IV in the Appendix. 
\begin{comment}
According to Eq.~\eqref{Eq_lmax_angle_range}, $[l_a.\theta,l_b.\theta]$ does not contain any node's critical angle values, which means that, when projected angle changes from $l_a.\theta$ to $l_b.\theta$, no new nodes enter into or leave out the charging cone, so there must be $\mathcal{N}_\text{Cone}(O,\theta_1,\phi){=}\mathcal{N}_\text{Cone}(O,\theta_2,\phi)$. 

Recall that a node on the boundary of a charging cone is regarded as covered by the charging cone, so the result still holds when either or both $\theta_1$ and $\theta_2$ are at the ends of the angle range $[l_a.\theta,l_b.\theta]$.     
\end{comment}
\end{proof}

\begin{lemma}
\label{lemma_lmax_angle_viseversa}
In context $\mathcal{C}_\text{ChrgPos}(O,A)$, for any projected angle $\theta_1{\in}[l_a.\theta,l_b.\theta]$ with $(l_a,l_b){\in}\mathcal{S}_\text{LMaxRng}(O,A)$, $\mathcal{N}_\text{Cone}(O,\theta_1,\phi)$ is an LMax-SCN set. Conversely, for any LMax-SCN set $s_1$ in context $\mathcal{C}_\text{ChrgPos}(O,A)$, there must be an item $(l_a,l_b){\in}\mathcal{S}_\text{LMaxRng}(O,A)$ in which, for any charging direction $\theta_1{\in}[l_a.\theta,l_b.\theta]$, there is $\mathcal{N}_\text{Cone}(O,\theta_1,\phi){=}s_1$.
\end{lemma}

\begin{proof}
Proof is provided in Sec.V in the Appendix. 
\begin{comment}
$\Longrightarrow$: Let inspect the variation of $\mathcal{N}_\text{Cone}(O,\theta_1,\phi)$ when projected angle $\theta_1{\in}[l_a.\theta,l_b.\theta]$ moves cross the endpoints $l_a.\theta$ and $l_b.\theta$. When $\theta_1$ reduces to cross $l_a.\theta$, as $l_a.\delta{=}0$, the node $l_a.\eta$ will leave out $\mathcal{N}_\text{Cone}(O,\theta_1,\phi)$. When $\theta_1$ increases cross $l_b.\theta$, as $l_b.\delta{=}1$, the node $l_b.\eta$ will leave out $\mathcal{N}_\text{Cone}(O,\theta_1,\phi)$. As a result, $\mathcal{N}_\text{Cone}(O,\theta_1,\phi)$ with $\theta_1{\in}[l_a.\theta,l_b.\theta]$ is an LMax-SCN set. 

$\Longleftarrow$: In context $\mathcal{C}_\text{ChrgPos}(O,A)$, as an LMax-SCN set, set $s_1$ must contain node $A$ with $A$ is on the boundary of a charging cone. Let denote the projected angle of the charging cone's direction as $\theta_1$, then $\theta_1$ must be contained in an angle range $[l_i.\theta,l_{(i{+}1)}.\theta]$ with $l_i.\delta{=}0$ and $l_{(i{+}1)}.\delta{=}1$. Otherwise we have $l_i.\delta{=}1$ or $l_{(i{+}1)}.\delta{=}0$. But in such cases, increasing $\theta_1$ to cross $l_{(i{+}1)}.\theta$ will additionally covers node $l_{(i{+}1)}.\eta$, whereas deceasing $\theta_1$ to cross $l_{i}.\theta$ will additionally covers node $l_{i}.\eta$, both implies that $s_1$ is not an LMax-SCN set, which contradicts with given condition that $s_1$ is an LMax-SCN set.
\end{comment}
\end{proof}

Fig.~\ref{fig_range_bars} illustrates the determination of 4 tuple items in $\mathcal{S}_\text{LMaxRng}(O,A)$, where each bar with light blue background represents an item $(l_a,l_b)$. In this figure, the four tuples items in $\mathcal{S}_\text{LMaxRng}(O,A)$ are $(l_3,l_4)$, $(l_5,l_6)$, $(l_8,l_9)$ and $(l_{11},l_{12})$, which correspond to LMax-SCN angle ranges of $[\theta_\text{C1},\theta_\text{B2}]$, $[\theta_\text{D1},\theta_\text{C2}]$, $[\theta_\text{F1},\theta_\text{D2}]$ and $[\theta_\text{G1},\theta_\text{G2}]$, respectively. Each of these angle ranges corresponds to an LMax-SCN node set. To be specific, the LMax-SCN node sets corresponding to these angle ranges are $\{B,C,E\}$, $\{C,D,E\}$, $\{D,F\}$ and $\{G\}$, respectively.
 
\begin{figure}[!htbp]
\centering
\includegraphics[width=0.99\linewidth]{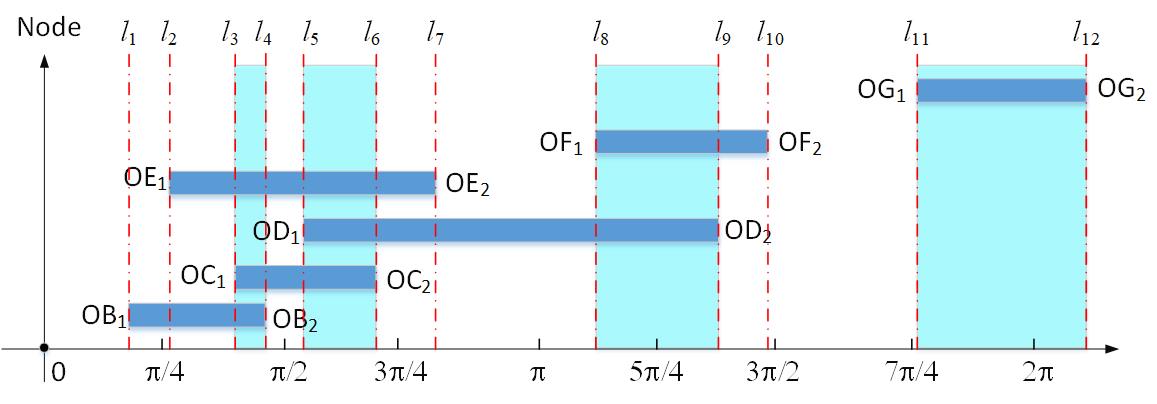}
\caption{Sketch map for determining LMax-SCN angle ranges in $\mathcal{S}_\text{LMaxRng}(O,A)$.}
\label{fig_range_bars}
\vspace{-0.6cm}
\end{figure}

\subsection{Select and Refine Representative Directions for Each LMax-SCN Angle Range}

As the projected angles in a $(l_a,l_b){\in}\mathcal{S}_\text{LMaxRng}(O,A)$ cover the same node set, the charging directions are FuncEqv with each other. As a basic idea, we want to make the covered nodes near to the enter-line of the charge cone as much as possible. Here we select a angle using Eq.~\eqref{Eq_angle_range_rep_equ} as the representation of the angle range $[l_a.\theta, l_b.\theta]$. Fig.~\ref{fig_repdir_alter} illustrates a projected angle $\theta_\text{REP}$ representing the angle range $[\theta_\text{C1},\theta_\text{B2}]$.  

\begin{equation}
\label{Eq_angle_range_rep_equ}
\frac{\theta_\text{REP}{-}l_a.\theta}{l_b.\theta-\theta_\text{REP}}{=}\frac{l_a.\tau}{l_b.\tau}
\end{equation}

\begin{figure}[!htbp]
\centering
\includegraphics[width=0.55\linewidth]{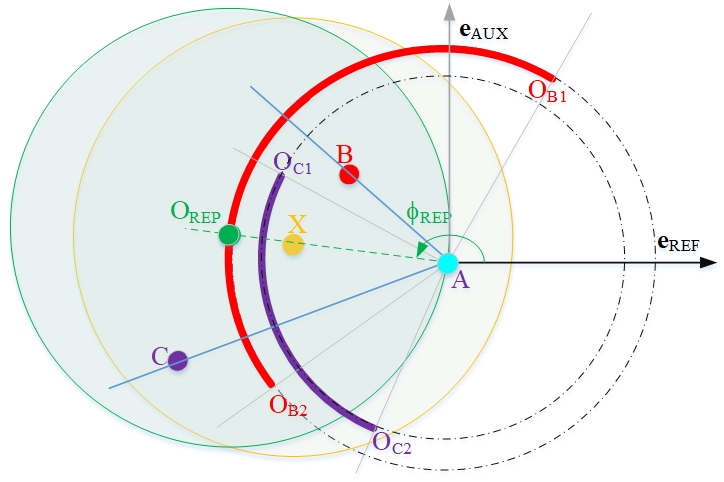}
\caption{Sketch map for selecting and refining the representative charging direction for an LMax-SCN angle range.}
\label{fig_repdir_alter}
\end{figure}

Eq.~\eqref{Eq_angle_range_rep_equ} leads to 

\begin{equation}
\label{Eq_angle_range_rep}
\theta_\text{REP}{=}\frac{(l_b.\theta){\cdot}(l_a.\tau){+}(l_a.\theta){\cdot}(l_b.\tau)}{(l_a.\tau){+}(l_b.\tau)}
\end{equation}

Eq.~\eqref{Eq_angle_range_rep_equ} as the representative of the angle range $[l_a.\theta, l_b.\theta]$, which yields the result in Eq.~\eqref{Eq_angle_range_rep}. Eq.~\eqref{Eq_angle_range_rep} is valid for tuples $[l_a.\theta, l_b.\theta]$ where $l_b.\theta {>} l_a.\theta$. When $l_b.\theta {<} l_a.\theta$, we can update $l_b.\theta {=} l_b.\theta + 2\pi$ as a preprocessing operation, and the resulting $\theta_\text{REP}$ obtained using Eq.~\eqref{Eq_angle_range_rep} should be subtracted by $2\pi$ to ensure $\theta_\text{REP} {\in} [0, 2\pi)$. 

\begin{comment}
    \begin{equation}
\label{Eq_angle_range_rep2}
\theta_\text{REP}{=}
\begin{cases}
\theta_\text{REP}-2\pi,&\text{if}\hspace{1mm}\theta_\text{REP}{\geq}2\pi; \\
\theta_\text{REP},&\text{otherwise}. 
\end{cases}
\end{equation}
\end{comment}

In context $\mathcal{C}_\text{ChrgPos}(O,A)$, the reference node $A$ is on the boundary of all charging cones with protected angles in range $[0,2\pi)$, including the selected $\theta_\text{REP}$. We denote the projection point corresponding to $\theta_\text{REP}$ as $O_\text{REP}$, which can be obtained using Eq.~\eqref{Eq_angle2_point} in Lemma~\ref{lemma_angle_range}. We then alter the direction $\theta_\text{REP}$ by moving node $A$ inward the charging cone, meanwhile ensuring that no node leaves out of the charging cone. To this end, we make a one dimensional search on the line segment $AO_\text{REP}$ from point $O_\text{REP}$ towards $A$ discretely, and select the one satisfying Eq.~\eqref{Eq_angle_rep_opt} as a final representative direction $\textbf{e}_\text{REP}{:=}\text{normalize}(\overrightarrow{OX})$ for the corresponding angle range.

\begin{subequations}
\label{Eq_angle_rep_opt}
\begin{align}
\arg{\hspace{0.2mm}}\min\limits_{\overrightarrow{OX}}{\hspace{0.2mm}}&\max\limits_{Y{\in}\mathcal{N}_\text{Cone}(O,\overrightarrow{OX},\phi)}
\angle{XOY}{\hspace{0.2mm}}{=}\arccos\left(\frac{\overrightarrow{OX}{\cdot}\overrightarrow{OY}}{|\overrightarrow{OX}||\overrightarrow{OY}|}\right), \label{Eq_angle_rep_opt1}\\
    \mathrm{s.t.}\quad  & X{\in}[O_\text{REP},A];\label{Eq_angle_rep_opt2}\\\
    &\mathcal{N}_\text{Cone}(O,\overrightarrow{OX},\phi){\supseteq}\mathcal{N}_\text{Cone}(O,\overrightarrow{OO_\text{REP}},\phi);\label{Eq_angle_rep_opt3}
\end{align}
\end{subequations}

Till now, for a charging point $O$ and reference node $A$, we obtain a set of representative directions with each covers an LMax-SCN set. Let $s_\text{REP}(O,A)$ be the set of representative directions corresponding to the LMax-SCN angle ranges in $\mathcal{S}_\text{LMaxRng}(O,A)$. For each $\textbf{e}_\text{REP}$ and its corresponding $\mathcal{S}_\text{Cone}(O,\textbf{e}_\text{REP},\phi)$, we create a tuple $(\textbf{v}_\text{Dir},S_\text{Cone})$ with $\textbf{v}_\text{Dir}{=}\textbf{e}_\text{REP}$ and $S_\text{Cone}{=}\mathcal{S}_\text{Cone}(O,\textbf{e}_\text{REP},\phi)$, and collect all such tuples into a set $\mathcal{S}_\text{DirRep}(O,A)$, i.e.,

\begin{align}
\label{Eq_construct_posref_dir}
\mathcal{S}_\text{DirRep}(O,A){:=}
\{(\textbf{e}_\text{REP},\mathcal{S}_\text{Cone}(O,\textbf{e}_\text{REP},\phi))| \textbf{e}_\text{REP}{\in}s_\text{REP}(O,A)\}. 
\end{align}
\vspace{-0.8cm}

\subsection{Build FuncEqv Representative Direction Set for a Charging Position}

Above analyses are conducted in the context $\mathcal{C}_\text{ChrgPos}(O,A)$.
For a certain charging position $O$, we use each node $A{\in}\mathcal{N}_\text{Sphere}(O)$ as a reference node to obtain its representative direction set $\mathcal{S}_\text{DirRep}(O,A)$, and combine them into a larger set using Eq.~\eqref{Eq_construct_pos_repset}.

\begin{equation}
\label{Eq_construct_pos_repset}
\mathcal{S}_\text{DirRep}(O){:=}\mathop{\cup}\limits_{A{\in}\mathcal{N}_\text{Sphere}(O)}\mathcal{S}_\text{DirRep}(O,A). 
\end{equation}

Each charging direction in $\mathcal{S}_\text{DirRep}(O,A)$ corresponds to an LMax-SCN set. This statement is valid in the context with a certain reference node, e.g., $\mathcal{C}_\text{ChrgPos}(O,A)$. In the larger context $\mathcal{C}_\text{ChrgPos}(O)$, it may become invalid. In $\mathcal{S}_\text{DirRep}(O)$, there may be some charging directions that are functional better than some other directions. In other words, $\mathcal{C}_\text{ChrgPos}(O)$ may contain some non-LMax-SCN directions. We name this issue as \textit{non-LMax-SCN direction issue}. The following lemma help us to refine $\mathcal{S}_\text{DirRep}(O)$ by checking and removing non-LMax-SCN direction items, as utilized in code line \ref{alg_line_remove_non_max} in Alg.~\ref{alg_cmfeds}.

\begin{lemma}
\label{lemma_contexto_nonlmax_remove}
In context $\mathcal{C}_\text{ChrgPos}(O)$, for any two nodes $A,B{\in}\mathcal{N}_\text{Sphere}(O)$, the non-LMax-SCN direction issue may happen between $\mathcal{S}_\text{DirRep}(O,A)$ and $\mathcal{S}_\text{DirRep}(O,B)$ if and only if $\varphi_{AOB}{\leq}\phi$. 
\end{lemma}

\begin{proof}
Proof is provided in Sec.VI in the Appendix. 
\begin{comment}
It is obvious that, if non-LMax-SCN direction issue, there must be $\varphi_{AOB}{\leq}\phi$. Now we focus on proving the reverse. We will prove it by showing that, if $\varphi_{AOB}{>}\phi$, non-LMax-SCN direction issue is impossible. In fact, all directions in $\mathcal{S}_\text{DirRep}(O,A)$ covers node $A$, whereas all directions in $\mathcal{S}_\text{DirRep}(O,B)$ covers $B$. However, as $\varphi_{AOB}{>}\phi$, $A$ and $B$ can not be simultaneously covered by any single charging cone, so node $B$ must not be covered by any directions in $\mathcal{S}_\text{DirRep}(O,A)$, and vise-versa. Hence no any direction in $\mathcal{S}_\text{DirRep}(O,A){\cup}\mathcal{S}_\text{DirRep}(O,B)$ can be functional better than any others, i.e., non-LMax-SCN direction issue do not exist. The lemma follows.

\end{comment}
\end{proof}

$\mathcal{S}_\text{DirRep}(O)$ is indeed a set of tuples that contains charging direction information, not directly containing the charging directions as its elements. However, $\mathcal{S}_\text{DirRep}(O)$ is mainly for storing the charging directions with other fields as auxiliary information for facilitating algorithm application, with a slight abuse of symbols, we also regard $\mathcal{S}_\text{DirRep}(O)$ as a charging direction set for expression simplicity, where the actual meaning can be identified from the context.   

Refining $\mathcal{S}_\text{DirRep}(O)$ will remove non-LMax-SCN direction items, making $\mathcal{S}_\text{DirRep}(O)$ optimal, as indicated by the following lemmas. 

\begin{lemma}
\label{lemma_contain_all_maxset}
The representative directions contained in $\mathcal{S}_\text{DirRep}(O)$ all correspond to LMax-SCN sets; For any LMax-SCN set $s_1$ at a certain charging position $O{\in}\mathcal{N}_\text{ChrgPos}$, there must be a tuple $(\textbf{v}_1,s_1){\in}\mathcal{S}_\text{DirRep}(O)$.
\end{lemma}

\begin{proof}
Proof is provided in Sec.VII in the Appendix. 
\end{proof}

\begin{lemma}
\label{lemma_func_equiv}
$\mathcal{S}_\text{DirRep}(O)$ is a minimum-size charging direction set that is FuncEqv to $\mathcal{S}_\text{Dir,Sphere}$, i.e., $\mathcal{S}_\text{Dir,Sphere}{\equiv}\mathcal{S}_\text{DirRep}(O)$.
\end{lemma}

\begin{proof}
Proof is provided in Sec.VIII in the Appendix. 
\begin{comment}
We first show that $\mathcal{S}_\text{Dir,Sphere}{\equiv}\mathcal{S}_\text{DirRep}(O)$.
According to Lemma~\ref{lemma_equiv_delete_worse}, removing all non-LMax-SCN directions leads to a set of angle ranges, say $\mathcal{S}_\text{DelNonLMax}$, such that $\mathcal{S}_\text{Dir,Sphere}{\equiv}\mathcal{S}_\text{DelNonLMax}$.

According to Lemma~\ref{lemma_contain_all_maxset}, any LMax-SCN node set is covered by some direction in $\mathcal{S}_\text{PosDir}$. Together with Lemma~\ref{lemma_equiv_select_representative}, we have $\mathcal{S}_\text{DelNonLMax}{\equiv}\mathcal{S}_\text{DirRep}(O)$. 

Chaining above results, we have $\mathcal{S}_\text{Dir,Sphere}{\equiv}\mathcal{S}_\text{DelNonLMax}$ ${\equiv}\mathcal{S}_\text{DirRep}(O)$. 

Lemma~\ref{lemma_contain_all_maxset} also show that, the representative directions contained in $\mathcal{S}_\text{DirRep}(O)$ all correspond to LMax-SCN sets. Thus, removing any direction will make an LMax-SCN set not covered by $\mathcal{S}_\text{DirRep}(O)$. Hence the size of $\mathcal{S}_\text{DirRep}(O)$ is minimized. 
The lemma follows.
\end{comment}
\end{proof}

\subsection{Generate Minimum-Size Pos-Dir Set}

We construct set $\mathcal{S}_\text{PosDir}$ following Eq.~\eqref{Eq_construct_dirrep_all}, where each tuple $(q.\textbf{v}_\text{Dir},O,q.S_\text{Cone})$ contains the information about the corresponding  charging direction, charging position and the set of covered nodes. Compared with the set $\mathcal{S}_\text{PosDir}$ mentioned in Sec.~\ref{s3_energy_cost_model_for_uav} in page \pageref{s3_energy_cost_model_for_uav}, the set of covered nodes $q.S_\text{Cone}$ is provided for facilitating the construction of energy transfer coefficient matrix $\textbf{C}_\text{ETC}$.

\begin{align}
\mathcal{S}_\text{PosDir}
{:=}&
\{(q.\textbf{v}_\text{Dir},O,q.S_\text{Cone})
|\nonumber\\
&q{\in}\mathcal{S}_\text{DirRep}(O,A),A{\in}\mathcal{N}_\text{Sphere}(O),O{\in}\mathcal{N}_\text{ChrgPos}\}
\label{Eq_construct_dirrep_all}
\end{align}

Based on above analyses, we design the cMFEDS algorithm for obtaining a minimum-size FuncEqv direction set, which are the final charging directions to be utilized. Its pseudo code is shown in Alg.~\ref{alg_cmfeds}. 

\begin{algorithm}[htb]
    \caption{Create a Minimum FuncEqv Direction Set (cMFEDS).}
    \label{alg_cmfeds}
    \LinesNumbered
    \KwIn  {node set $\mathcal{N}$, charging position set $\mathcal{N}_\text{ChrgPos}$, charge cone angle $\phi$, charging distance $D$;}
    \KwOut{Pos-Dir set $\mathcal{S}_\text{PosDir}$;}
    
    \For{(Each position $O{\in}\mathcal{N}_\text{ChrgPos}$)}
    { \label{alg_line_for_pos}
        Normalize the nodes in $\mathcal{N}_\text{Sphere}(O)$ using Eq.~\eqref{Eq_node_proj}.\label{alg_line_normalize_nodes}\\
        \For {(\text{Each node} $A{\in}\mathcal{N}_\text{Sphere}(O)$)}{\label{cmfrds_for_begin}
             \If{($\mathcal{N}_\text{Cone}(O,\overrightarrow{OA},2\phi){=}\{A\}$)}{
             $\mathcal{S}_\text{DirRep}(O,A){=}\{(\textbf{e}_\text{OA},\{A\})\}$;\\
             \textbf{Continue;}
            }
            \For {(\text{Each node} $B{\in}\mathcal{N}_\text{Cone}(O,\overrightarrow{OA},2\phi))$ }
            { \label{alg_line_conref_begin}
                Obtain $\theta_\text{B1},\theta_\text{B2}$ using Eq.~\eqref{Eq_proj_angle};\label{alg_line_cal_endpoints}\\
                construct tuples using Eq.~\eqref{Eq_endpoint_tuple};\label{alg_line_construct_tuple}\\
            }\label{alg_line_conref_end}
            Sort the tuples in ascending order of angles;\label{alg_line_sort_angles}\\
            Create $\mathcal{S}_\text{LMaxRng}(O,A)$ using Eq.~\eqref{Eq_lmax_angle_range};\label{alg_line_create_lmaxrng}\\ 
            Calculate and refine $\theta_\text{REP}$ directions using Eq.~\eqref{Eq_angle_range_rep_equ}, and Eq.~\eqref{Eq_angle_rep_opt};\label{alg_line_create_refine_dir_rep}\\
            Construct $\mathcal{S}_\text{DirRep}(O,A)$ using Eq.~\eqref{Eq_construct_posref_dir};\label{alg_line_create_dirset_oa}  
        }\label{cmfrds_for_end}
        Construct $\mathcal{S}_\text{DirRep}(O)$ using Eq.~\eqref{Eq_construct_pos_repset};\label{alg_line_create_dirset_o} \\ 
        \For {$(A,B{\in}\mathcal{N}_\text{Sphere}(O))$}
        {\label{alg_line_remove_non_max_loop_begin}
            \If{$(\varphi_{AOB}{\leq}\phi)$}{
                Check and remove non-LMax-SCN directions in $\mathcal{S}_\text{DirRep}(O)$;\label{alg_line_remove_non_max}
            }
        }\label{alg_line_remove_non_max_loop_end}
    }
    Construct $\mathcal{S}_\text{PosDir}$ using Eq.~\eqref{Eq_construct_dirrep_all};\label{alg_line_s_dirrep_final}\\
    \textbf{return} $\mathcal{S}_\text{PosDir}$;
\end{algorithm}

\begin{theorem}
\label{theorem_alg_optimal}
cMFEDS optimally solves the P2 problem and returns a set $\mathcal{S}_\text{PosDir}$ with minimum size.
\end{theorem}

\begin{proof}
Lemma~\ref{lemma_func_equiv} shows that, for each charging position, the algorithm solves a corresponding instance of the P2 problem and returns an optimal solution $\mathcal{S}_\text{DirRep}(O)$ with minimum size. Furthermore, as the charging directions at different positions in $\mathcal{N}_\text{ChrgPos}$ are not comparable, the set $\mathcal{S}_\text{PosDir}$ obtained by unifying $\mathcal{S}_\text{DirRep}(O)$ with $O{\in}\mathcal{N}_\text{ChrgPos}$ is also minimum.    
\end{proof}

\begin{lemma}
\label{lemma_cMFEDS_timecomplexity}
The time complexity required by cMFEDS is $\mathcal{O}(P{\times}N^2)$, where $P$ represents the number of charging positions, and $N$ represents the number of nodes.
\end{lemma}
\begin{proof}
Proof is provided in Sec.IX in the Appendix. 
\begin{comment}
As shown in Alg.~\ref{alg_cmfeds}, cMFEDS is mainly a for-loop containing $P$ loops each for one charging position. In each loop, the node normalization operation (line \ref{alg_line_normalize_nodes}) requires $\mathcal{O}(N)$ time. The nested inner for-loop (lines \ref{cmfrds_for_begin}-\ref{cmfrds_for_end}) traverses each node in $\mathcal{N}_\text{Sphere}$ and each node in $\mathcal{N}_\text{Cone}$ within the cone, requiring a time complexity of  $\mathcal{O}(|\mathcal{N}_\text{Sphere}|{\times}|\mathcal{N}_\text{Cone}|){=}\mathcal{O}(N^2)$. Sorting the tuples in line~\ref{alg_line_sort_angles} has a time complexity of $\mathcal{O}(|\mathcal{N}_\text{Cone}|{*}\log(|\mathcal{N}_\text{Cone}|)){=}\mathcal{O}(N{*}\log{N})$. Then, constructing and refining the direction set (lines \ref{alg_line_conref_begin}-\ref{alg_line_conref_end}) requires $\mathcal{O}(|N_\text{Sphere}|){=}\mathcal{O}(N)$ time. Checking and deleting node pairs (lines \ref{alg_line_remove_non_max_loop_begin}-\ref{alg_line_remove_non_max_loop_end}) has a complexity of $\mathcal{O}(|N_\text{Sphere}|^2){=}\mathcal{O}(N^2)$. Combining all these results, the overall time complexity of the out-most for-loop is $\mathcal{O}(P{\times}(N{+}N^2{+}N{*}{\log}{N}{+}N{+}N^2)){=}\mathcal{O}(PN^2)$. Together with the additional operation of constructing direction set (line \ref{alg_line_s_dirrep_final}) has a complexity of $\mathcal{O}(P)$, the time complexity of cMFEDS is $\mathcal{O}(PN^2{+}P){=}\mathcal{O}(PN^2)$.  
\end{comment}
\end{proof}

\section{The FELKH-3D Algorithm}
\label{sec_algorithm}

We propose FELKH-3D to address the DCS-3D problem in three steps, as previously outlined in Fig.~\ref{fig_DCSBGS} in Sec.~\ref{sec_intro}.
\begin{comment}
   \begin{enumerate}[\textbf{S}1:]
\item Construct the Pos-Dir pair set $\mathcal{S}_\text{PosDir}$ by using the cMFEDS algorithm  in Sec.~\ref{sec_address_direction_challenge}.
\item Determine the energy transmission time along the Pos-Dir pairs in $\mathcal{S}_\text{PosDir}$.
\item Find an energy efficient loop tour using LKH algorithm\cite{lkh_2019} to traverse the charging positions, bypassing the positions with total energy transmission time equals 0.  
\end{enumerate} 
\end{comment}

\subsection{Determine Charging Times along Pos-Dir Pairs}
\label{s6_charging_time_determination}
Let $\textbf{t}(\textbf{s}){:=}[t_1,t_2,\ldots,t_{K}]^{\text{T}}$ denotes the column vector of the energy transmission times in a charging schedule $\textbf{s}$ along the Pos-Dir pairs in $\mathcal{S}_\text{PosDir}$. As UAV's flying energy consumption is not affected by $\textbf{t}$, when only charging times as used as optimization variables, minimizing total energy loss $e^\text{Total}_\text{Loss}(\textbf{s})$ is equivalent to minimize $e^\text{WPT+Hov}_\text{Loss}$. Thus, the optimal $\textbf{t}^*$ can be determined by solving P3 in Eq.~\eqref{eq_energy_step2_posdirtime}.

\begin{subequations}
\label{eq_energy_step2_posdirtime}
\begin{alignat}{2}
\textbf{(P3)}\quad 
&\textbf{t}^* = &&\arg\min_{\textbf{t}} \; e^{\text{WPT+Hov}}_{\text{Loss}}(\textbf{t}),\\
&\text{s.t.}\; 
&&~Eq.\eqref{eq_DWRACS_problem_math_C1},~Eq.\eqref{eq_DWRACS_problem_math_C2}; \nonumber\\
&&&\textbf{t} \ge 0.
\end{alignat}
\end{subequations}

Problem P3 can be solved using mature softwares such as \textbf{Cplex}. Combining $\textbf{t}^*$ with the Pos-Dir pairs in set $\mathcal{S}_\text{PosDir}$, the flying schedule item set $\mathcal{S}^\text{CTS}_\text{Fly}$ can be constructed easily. 

\subsection{Determining UAV's Trajectory}
\label{s5_solve_the_uav_path_problem}
Having obtained $\textbf{t}^*$ in the previous step, let $t_\text{Pos}(i)$ denote the sum of the times corresponding to the Pos-Dir pairs with position $l_i$. If $t_\text{Pos}(i){=}0$, then charging position $l_i$ need not be traversed. Let $\mathcal{L}^*_\text{Refine}{:=}\{l_i|t_\text{Pos}(i){\geq}0,u_i{\in}\mathcal{U}\}$, then the objective of this step is to determine a loop tour $\textbf{r}$ with minimum flying energy consumption to traverse all charging positions in $\mathcal{L}^*_\text{Refine}$. 

Let $\textbf{r}_\pi := [l{\pi_0} = l_0, l_{\pi_1}, \ldots, l_{\pi_m}, l_0]$, and let ${\textbf{r}_\pi}$ represent the set of locations contained in the path $\textbf{r}_\pi$, and $|\textbf{r}_\pi|$ represent the length of the path list, i.e., $|\textbf{r}_\pi| = m+2$. Let $e^\text{UAV}\text{Fly}(\textbf{r}\pi)$ represent the flight energy consumption of the UAV along the path $\textbf{r}_\pi$. The UAV traverses each location sequentially according to a predetermined visiting order. Therefore, the path planning task can be formalized as the formula ~\eqref{eq_ADMCCS_step4_TSP}, where constraint ~Eq.\eqref{eq_ADMCCS_step4_TSP_C7} ensures that each position in the set $\mathcal{L}^*_\text{Refine} \cup {l_0}$ is visited only once.

\begin{subequations}
\label{eq_DMCCS_step4_TSP}
\begin{align}
\textbf{(P4)} \quad 
& \textbf{r}_\pi^* = \arg\min_{\textbf{r}_\pi} \; e^{\text{UAV}}_{\text{Fly}}(\textbf{r}_\pi) \\
\text{s.t.} \quad 
& \{\textbf{r}_\pi\} = \mathcal{L}^*_\text{Refine} \cup \{l_0\}.\label{eq_ADMCCS_step4_TSP_C7}
\end{align}
\end{subequations}

 To solve P4, we exploit the most efficient and powerful heuristic algorithm named as LKH, which achieves the state-of-the-art performance both in solution quality and running speed. Performance evaluations well demonstrates the superiority of LKH. 

Given a loop route $\textbf{r}$, the flying schedule item set $\mathcal{S}_\text{Fly}^\text{CTS}$ can be constructed easily. A complete CTS schedule $\textbf{s}^\text{CTS}$ solving the Directional Mobile Charger Charge Scheduling (DMCCS) problem can then be obtained by merging $\mathcal{S}_\text{Fly}^\text{CTS}$ and $\mathcal{S}_\text{Chrg}^\text{CTS}$ and sorting the items following $\textbf{r}$. 

\subsection{The FELKH-3D Algorithm}
\label{s5_the_DCSBGS_algorithm}
Based on above analyses, the FELKH-3D algorithm as outlined in Alg.~\ref{alg_felkh} is obtained easily.

      \begin{algorithm}[htb]
		\caption{FELKH-3D}
         \label{alg_felkh}
		\LinesNumbered
        \KwIn{$\mathcal{N}$,$\phi$,$\text{D}$,$\textbf{e}_\textbf{E}$,
            $\textbf{e}_\textbf{B}$,$\textbf{p}_\textbf{hov}$,
             $\textbf{p}_\textbf{0}$,$\textbf{p}_\textbf{U}$,$p_\text{mov}$,$v$}
            \KwOut{$\text{A charging schedule }\textbf{s}_\text{CS}$}
            $\mathcal{S}_\text{PosDir}{=}$\text{cMFEDS}$(\mathcal{N},\mathcal{N}_\text{ChrgPos}{=}\mathcal{N},\phi,D)$;\label{alg_felkh:1}\\
            $\text{Construct ETC matrix }\mathbf{C}_\text{ETC}{:=}[c(i,j)]_{i{\in}\mathcal{S}_\text{PosDir},u_j{\in}\mathcal{U}}$;\label{alg_felkh:2}\\
            $[\textbf{t},e^\text{UAV}_\text{Hov},e^\text{UAV}_\text{Chrg},e^\text{Nodes}_\text{Rcv}]{=}\text{sloveP3}(\textbf{C}_\textbf{ETC},\textbf{e}_\textbf{B},\textbf{e}_\textbf{E},\textbf{p}_\textbf{0},\textbf{p}_\textbf{hov})$;\label{alg_felkh:3}\\
            Construct charging schedule item set $\mathcal{S}^\text{CTS}_\text{Chrg}$;\label{alg_felkh:4}\\
            Construct $\mathcal{L}^*_\text{Refine}{:=}\{l_i|t_\text{Pos}(i){\geq}0,u_i{\in}\mathcal{U}\}$ using $\textbf{t}$ and $\mathcal{S}_\text{PosDir}$;\label{alg_felkh:5}\\             $[\textbf{r},e^\text{UAV}_\text{Fly},\textbf{t}^\text{UAV}_\textbf{Fly}]=\text{LKH}(\mathcal{L}^*_\text{Refine})$;\label{alg_felkh:6}\\
            $e^\text{Total}_\text{Loss}=e^\text{UAV}_\text{Fly}{+}e^\text{UAV}_\text{Hov}{+}e^\text{UAV}_\text{Chrg}{-}e^\text{Nodes}_\text{Rcv}$;\label{alg_felkh:7}\\
            Construct flying schedule item set $\mathcal{S}^\text{CTS}_\text{Fly}$;\label{alg_felkh:8}\\
            $\textbf{s}^\text{CTS}{=}$combineAndSort$(\mathcal{S}^\text{CTS}_\text{Chrg},\mathcal{S}^\text{CTS}_\text{Fly})$;\label{alg_felkh:9}\\
            $\textbf{return } \textbf{s}^\text{CTS}$;
	  \end{algorithm}

\begin{lemma}
\label{lemma_timecomplexity_felkh3d}
The time complexity of FELKH-3D is $O(PN^2{+}MN{+}P^2)$, where $M{=}|\mathcal{S}_\text{PosDir}|$.
\end{lemma}
\begin{proof}
Proof is provided in Sec.X in the Appendix.
\begin{comment}
    The position selection operation in line~\ref{alg_felkh:1} is performed greedily at each node location, which has a time complexity of $\mathcal{O}(1)$. cMFEDS called in this step has a time complexity of $\mathcal{O}(PN^2)$. Constructing $\textbf{C}_{ETC}$ in line~\ref{alg_felkh:2} has a time complexity of $\mathcal{O}(NM)$. Additionally, solving the linear programming problem with $\textbf{Cplex}$ in line~\ref{alg_felkh:3} has a time complexity of $O(M)$. LKH in line~\ref{alg_felkh:6} has a time complexity of $\mathcal{O}(P^2)$. Combining above results, the overall time complexity of FELKH-3D is $\mathcal{O}(PN^2{+}MN{+}P^2)$.  
\end{comment}
\end{proof}
\section{Testbed experiment}
\label{sec_testbed}
We set up a small-scale 3D WRSN scenario outdoors and conducted tests on the FELKH-3D algorithm. 
\subsection{Experiment Setup}
As illustrated in Fig.~\ref{fig_intro}, the UAV is equipped with a Powercast TX91501 module for wireless power transfer, while five WRSN nodes employ Powercast P2110-EVB modules with energy storage. These nodes are deployed within a 4m × 4m × 2m three-dimensional space and periodically report their energy levels to a PC via wireless communication for centralized data processing. Although the TX91501 module supports a maximum transmission distance of 12m, the effective communication range is restricted to 2m to improve transmission efficiency. In the simulated deployment, nodes are positioned at A (1.5, 0.3, 1.2), B (2.7, 0.9, 0.9), C (2.7, 1.2, 0.8), D (3.3, 1.2, 0), and E (3.3, 1.8, 0), with the UAV base at (4, 4, 0).
Directional charging is achieved through the TX91501's built-in directional antenna, which features a 60° beam width and height. Mounted on the UAV, this module enables controlled directional charging by adjusting the drone’s pitch, roll, and yaw, ensuring precise alignment between the antenna’s emission direction and the target node. This mechanism enhances wireless power transfer efficiency.

\begin{figure}
    \centering
    \includegraphics[width=1\linewidth]{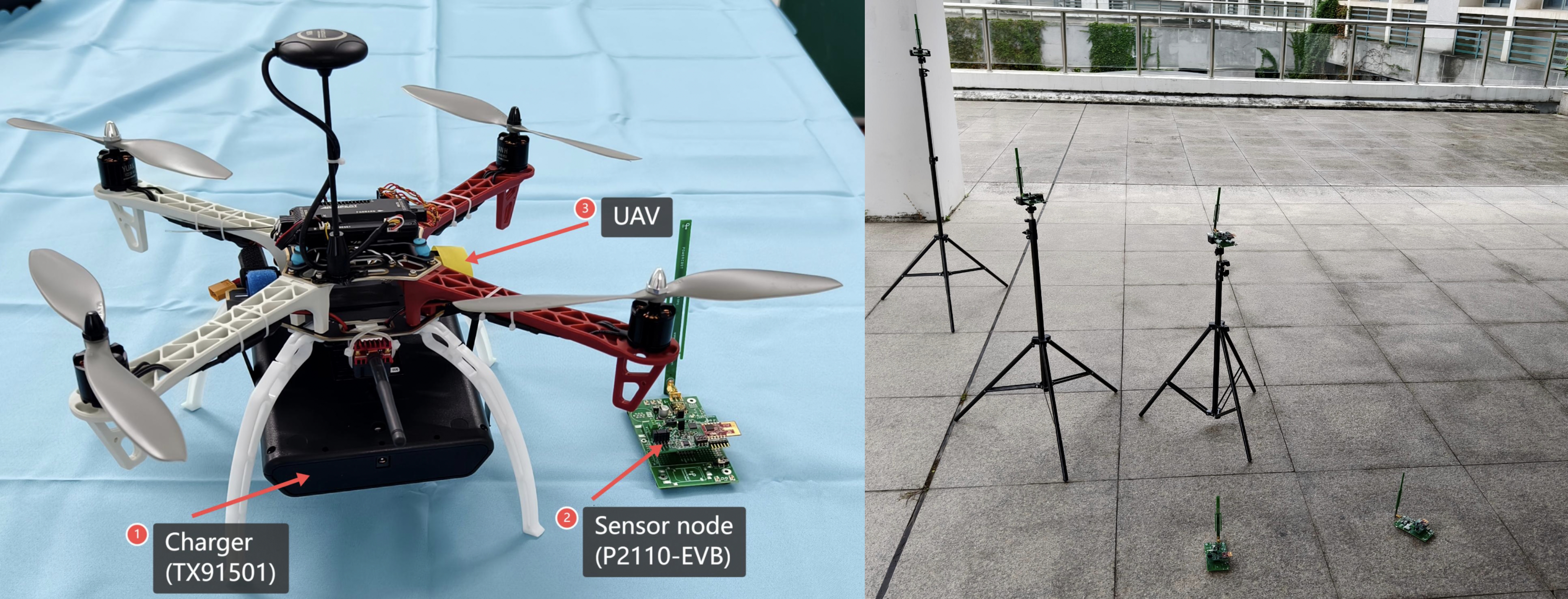}
    \caption{Test-bed devices and system}
    \label{fig_intro}
    \vspace{-0.8cm}
\end{figure}

\begin{figure}[!htbp]
    \centering
    % 第一行
    \begin{subfigure}[b]{0.48\linewidth} % 调整宽度
      \centering
    \includegraphics[width=1\linewidth]{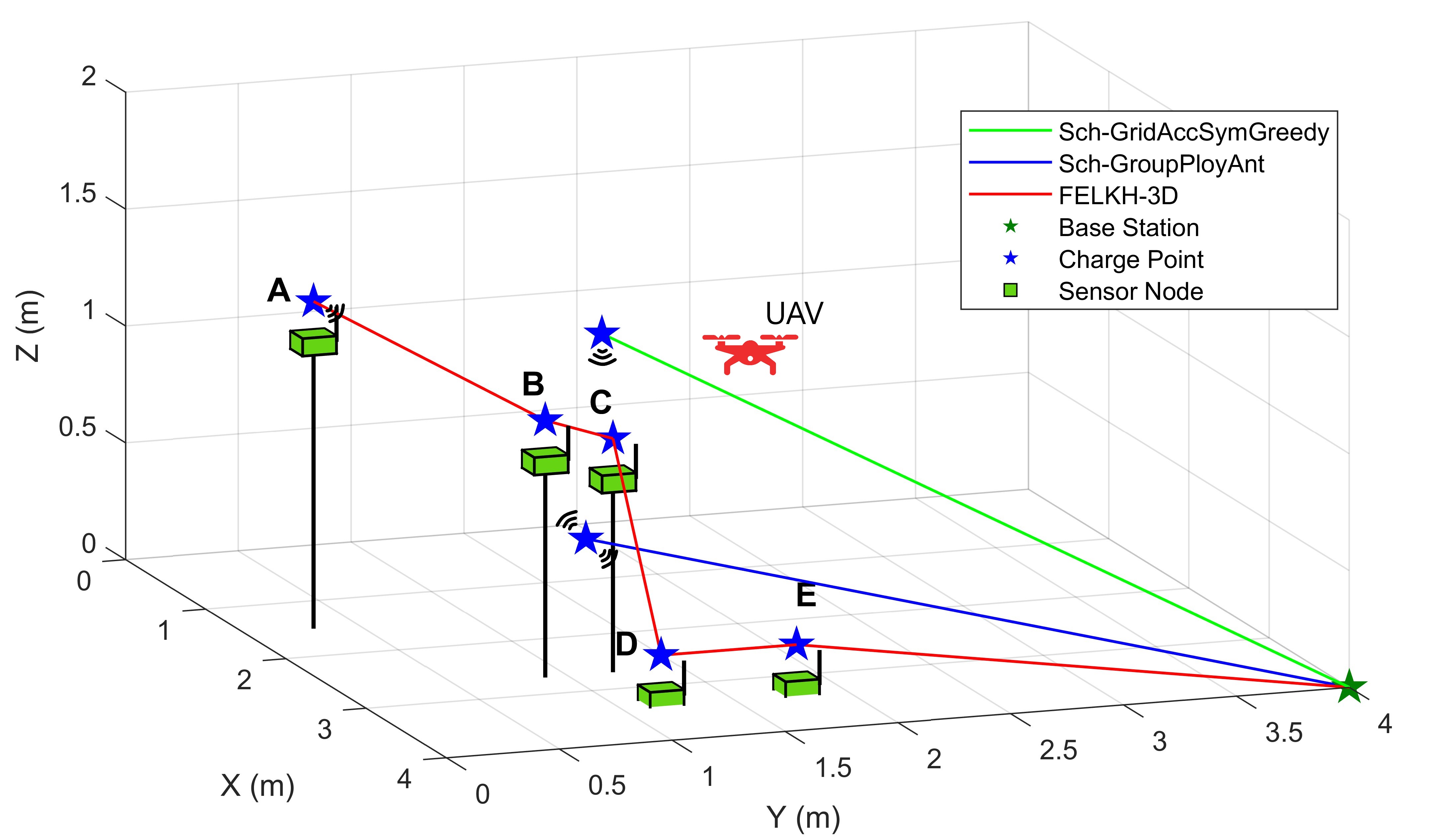}
    \caption{Test-bed scenario}
    \label{fig_test_scenario}
    \end{subfigure}
    \hfill
    % 第二行
    \begin{subfigure}[b]{0.48\linewidth} % 调整宽度
        \centering
        \includegraphics[width=1\linewidth]{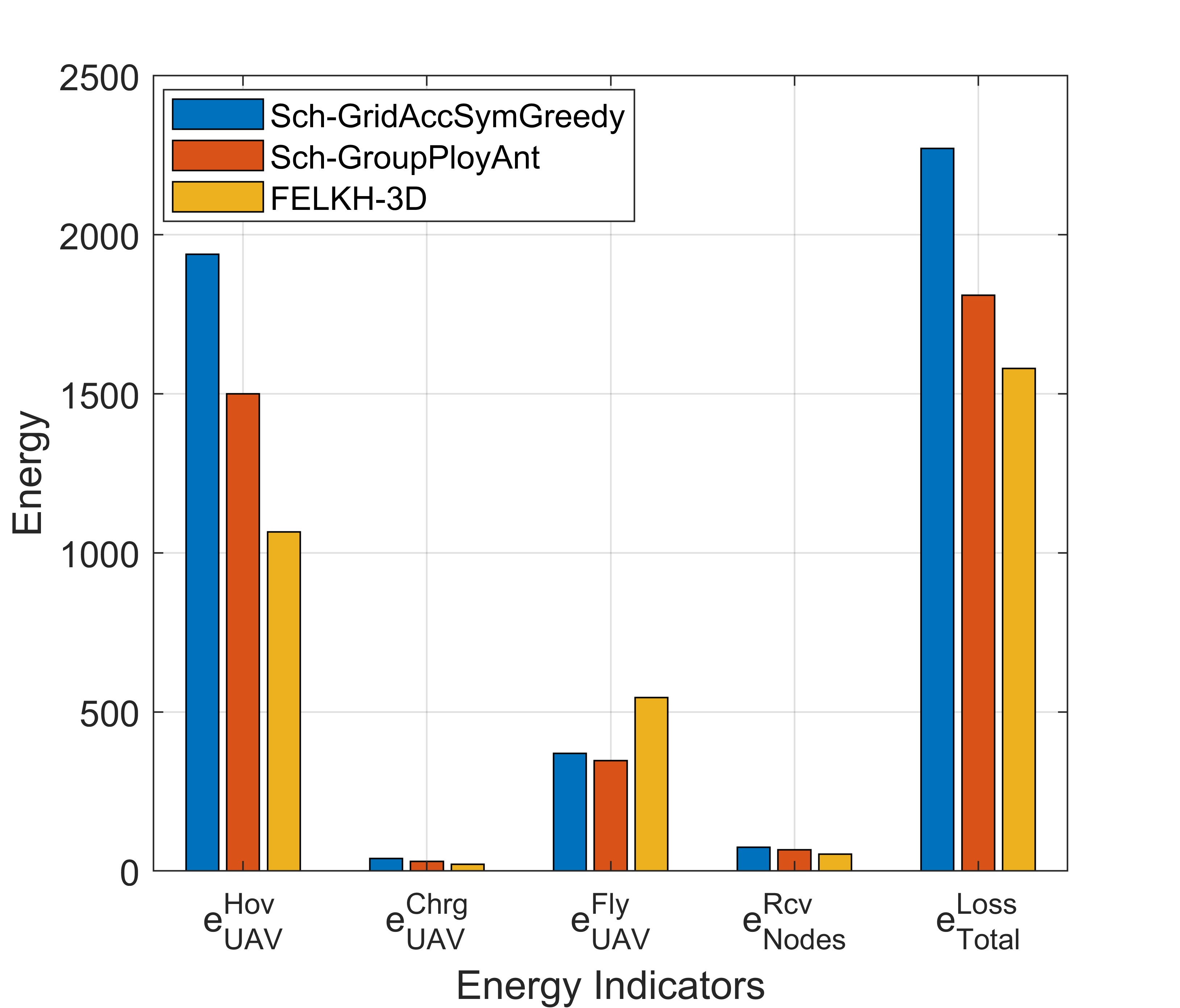}
        \caption{Experimental Results}
        \label{fig_test_results}
    \end{subfigure}
    \caption{Testbed experiment}
    \label{fig_res4}
    \vspace{-0.5cm}
\end{figure}

\begin{table}[htbp]
    \centering
    \caption{Experiment Parameters}
    \label{t_exp_simu}
    \small
    \begin{tabular}{c|c||c|c||c|c}
        \hline
        \makebox[0.05\textwidth][c]{\text{Parameter}} & \makebox[0.025\textwidth][c]{\text{Value}} & \makebox[0.05\textwidth][c]{\text{Parameter}} & \makebox[0.035\textwidth][c]{\text{Value}} & \makebox[0.05\textwidth][c]{\text{Parameter}} & \makebox[0.020\textwidth][c]{\text{Value}} \\
        \hline
        $e_\text{B}$ & 20$\sim$80 J & $e_\text{E}$ & 30$\sim$90 J & $n$ & 5\\ \hline
        $p_\text{0}$ & 3 W & $p_\text{Hov}$& 150 J/s& $p_\text{Fly}$& 160 J/s \\ \hline
        $l_0$ & (4,4,0) &  $D$ & 2 m & $v$ & 3 m/s \\ \hline
       \multicolumn{2}{c|}{Region} &  \multicolumn{4}{|c}{4m×4m×2m}\\ \hline
    \end{tabular}
    \vspace{-0.5cm}
\end{table}

\subsection{Experimental Results}
Figure~\ref{fig_test_scenario} illustrates the deployment configuration of five sensor nodes in a three-dimensional experimental scenario, along with the predefined flight paths and the WRSN. Three algorithms were compared, with results shown in Fig.\ref{fig_test_results}. The proposed FELKH-3D algorithm achieves the lowest energy loss—30\% and 13\% lower than the Sch-GridAccGreedy and Sch-GroupPloyAnt algorithms, respectively. This improvement is due to FELKH-3D’s ability to position the UAV closer to nodes during hovering. In contrast, Sch-GridAccGreedy constrains the UAV to a fixed altitude, reducing energy reception efficiency and yielding suboptimal charging directions. Sch-GroupPloyAnt forms charging groups based on node proximity, but suffers from prolonged charging times due to limited transmission power. FELKH-3D addresses both reception efficiency and power constraints by distributing charging time across multiple positions and applying efficient path planning to approximate an optimal route, thereby reducing charging cost and overall energy loss. Although FELKH-3D incurs slightly higher flight costs, it significantly reduces hovering cost—the dominant component—leading to minimal total energy loss, as shown in Fig.~\ref{fig_test_results}. In larger-scale scenarios with higher energy demands, faster UAV speeds, and stronger transmission capabilities, the advantage of FELKH-3D is expected to be even more substantial.

\section{Performance Evaluation}
\label{sec_simulation}
We evaluate FELKH-3D’s performance via numerical simulations using MATLAB 2022b running on a computer with
i5-9300H CPU, 16 GB RAM, and Windows 10 OS.

\subsection{Algorithms for Comparison}
\label{s7_algorithms_for_comparision}
As solving the DCS-3D problem involves several steps each has several alternative options/methods, we will comparatively evaluate the alternative options in the context of a complete algorithm. To emphasize the comparison of the methods for a step, all other steps take the same options. We also evaluate several complete algorithms with different step-wise options.

For charging position generation, the node-based method, Grid-Based method~\cite{Jiang}, and Clustering-based method~\cite{Liang2023TVT} are selected for comparison, which are denoted by ``\emph{Pos-Node}``, ``\emph{Pos-Grid}``, ``\emph{Pos-Group}``, and ``\emph{Pos-Cluster}``, respectively. For charging direction selection, we select the simple node-based direction, GCC~\cite{Jiang}, ACC~\cite{Jiang}, and our cMFEDS algorithm. As an representation of fixed direction method, we propose to divide the whole sphere surface into regular faces by referring to suitable regular polyhedrons or the soccer ball, and select the directions passing the centers of the faces as charging directions. The methods are denoted as ``\emph{Dir-Node}``, ``\emph{Dir-GCC}``, ``\emph{Dir-ACC}``, ``\emph{Dir-FuncEqv}``, and ``\emph{Dir-Polyhedron}``, respectively. For charging tour searching, 
 the ant colony algorithm was chosen as the representative heuristic algorithm for solving the TSP. In addition, classic TSP algorithms (LKH algorithm and greedy algorithm) were also used. Therefore, the three methods used are represented as "Tour-LKH", "Tour-Ant", and "Tour-Greedy" respectively.
Combining different options of the steps leads to complete algorithms for solving the DCS-3D problem. We will select typical complete algorithms to compare with our FELKH-3D in Sec.~\ref{sec_cmp_complete_schemes}.

\subsection{Performance Metrics and Simulation Setup}
\label{s7_preformace_metrics_and_simulation_setup}
Two primary performance metrics adopted here are \emph{total energy loss} and \emph{time span}. Detailed metrics such as charging energy consumption, charging time, flying time will also be checked when suitable. 

Table~\ref{t_simu} presents the crucial simulation parameters and the corresponding default values in our simulations. A typical set of parameter values is referred to as a simulation configuration here. The impact of a parameter on the algorithms is investigated by exhaustive experiments under various simulation configurations. The UAV’s energy consumption related parameters in the simulation experiments are referenced from~\cite{Lin2021TMC,sun2023CC,Tianle2023CC}. To mitigate randomness, 100 problem instances for each simulation configuration are tested and the results are averaged. We collect and average the performance metrics of the instances as the final result. 

\begin{table}[htbp]
    \centering
    \caption{Simulation Parameters}
    \label{t_simu}
    \small
    \begin{tabular}{c|c||c|c||c|c}
        \hline
        \makebox[0.05\textwidth][c]{\text{Parameter}} & \makebox[0.025\textwidth][c]{\text{Value}} & \makebox[0.035\textwidth][c]{\text{Parameter}} & \makebox[0.035\textwidth][c]{\text{Value}} & \makebox[0.045\textwidth][c]{\text{Parameter}} & \makebox[0.020\textwidth][c]{\text{Value}} \\
        \hline
        $e_\text{B}$ & 20$\sim$90 J & $e_\text{E}$ & 20$\sim$90 J & $e_\text{U0}$ & 90 J\\ \hline
        $\beta$ & 4 & $\gamma$ & 2 & $\delta$ & 12 \\ \hline
        $p_\text{0}$ & 1 W & $p_\text{Hov}$/$p_\text{Fly}$ & 0.1$\sim$1& $p$& 8W \\ \hline
        $l_0$ & (0,0,0) & $\phi$ & $\pi/3$& $v$ & 1 m/s \\ \hline
        $D$ & 6 m & $c_\text{max}$ & 0.9 & $n$ & 400$\sim$800 \\ \hline
       \multicolumn{2}{c|}{Region} &  \multicolumn{4}{|c}{100m×100m×20m}\\ \hline
    \end{tabular}
    \label{table2}
    \vspace{-0.6cm}
\end{table}

\subsection{Comparison of Charging Position Generation Methods}
\label{s7_performace_CPG}
We fix the options of Dir-FuncEqv and Tour-LKH and comparatively evaluate the four charging position generation methods with respect to node number $n$ and $p_\text{Hov}/p_\text{Fly}$.

1) Number of Nodes: We varied the number of nodes $n$ from 400 to 800 while keeping the ratio $p_\text{Hov}/p_\text{Fly}$ fixed at 1. The results are shown in Fig.~\ref{fig_res1}. Total energy loss, time span, charging position number, and tour length are used as performance metrics. 

Fig.~\ref{fig_res1_energy} shows that $e^\text{Total}_\text{Loss}$ of Pos-Node is lower than those of the other three algorithms. This is due to the fact that the energy transmission coefficient decreases as charging distance increases, resulting in higher energy consumption for Pos-Cluster, Pos-Group, and Pos-Grid, where the nodes are charged over a distance. This not only leads to higher charging energy consumption, but also results in charging schedules with longer time span. As shown in Fig.~\ref{fig_res1_time}, Pos-Node exhibits a more favorable time span. Fig.~\ref{fig_res1_count} shows the trend of charging position number as $n$ increases. Pos-Cluster and Pos-Group methods creates much fewer charging positions than Pos-Grid and Pos-Node, however fewer charging positions do not necessarily leads to lower energy loss schedules. In Fig.~\ref{fig_res1_tour}, Pos-Cluster and Pos-Group generate charging tours considerably shorter than the others, owing to their fewer charging positions. Although Pos-Node leads to longer tour distance, it achieves the least energy consumption and shortest time span, and so it is adopted in our FELKH-3D.    

\begin{figure*}[!htb]
    \centering
    \begin{subfigure}[b]{0.24\linewidth} % 调整宽度
        \centering
        \includegraphics[width=\linewidth]{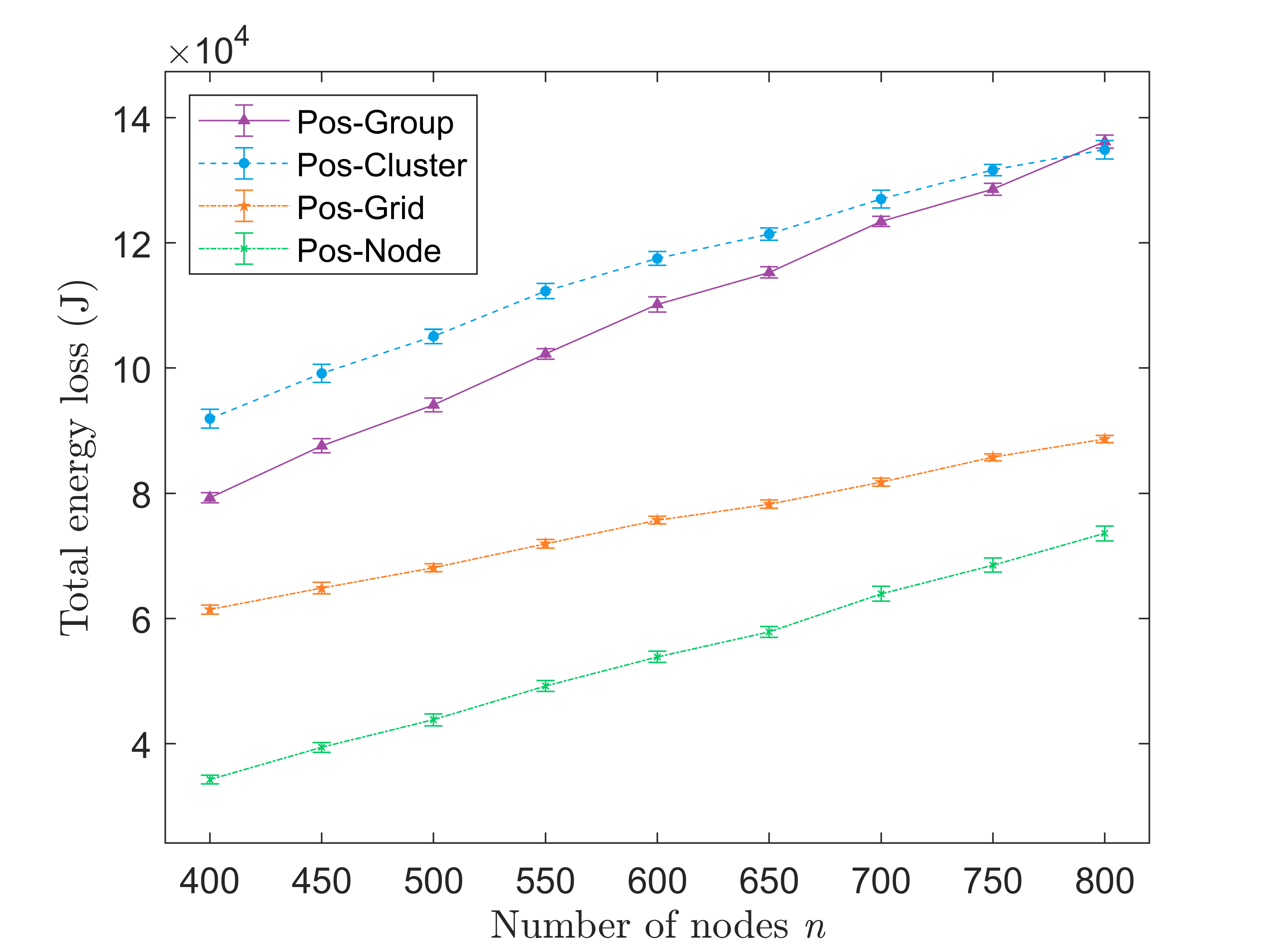}
        \caption{Total energy loss}
        \label{fig_res1_energy}
    \end{subfigure}
    \hfill
    \begin{subfigure}[b]{0.24\linewidth} % 调整宽度
        \centering
        \includegraphics[width=\linewidth]{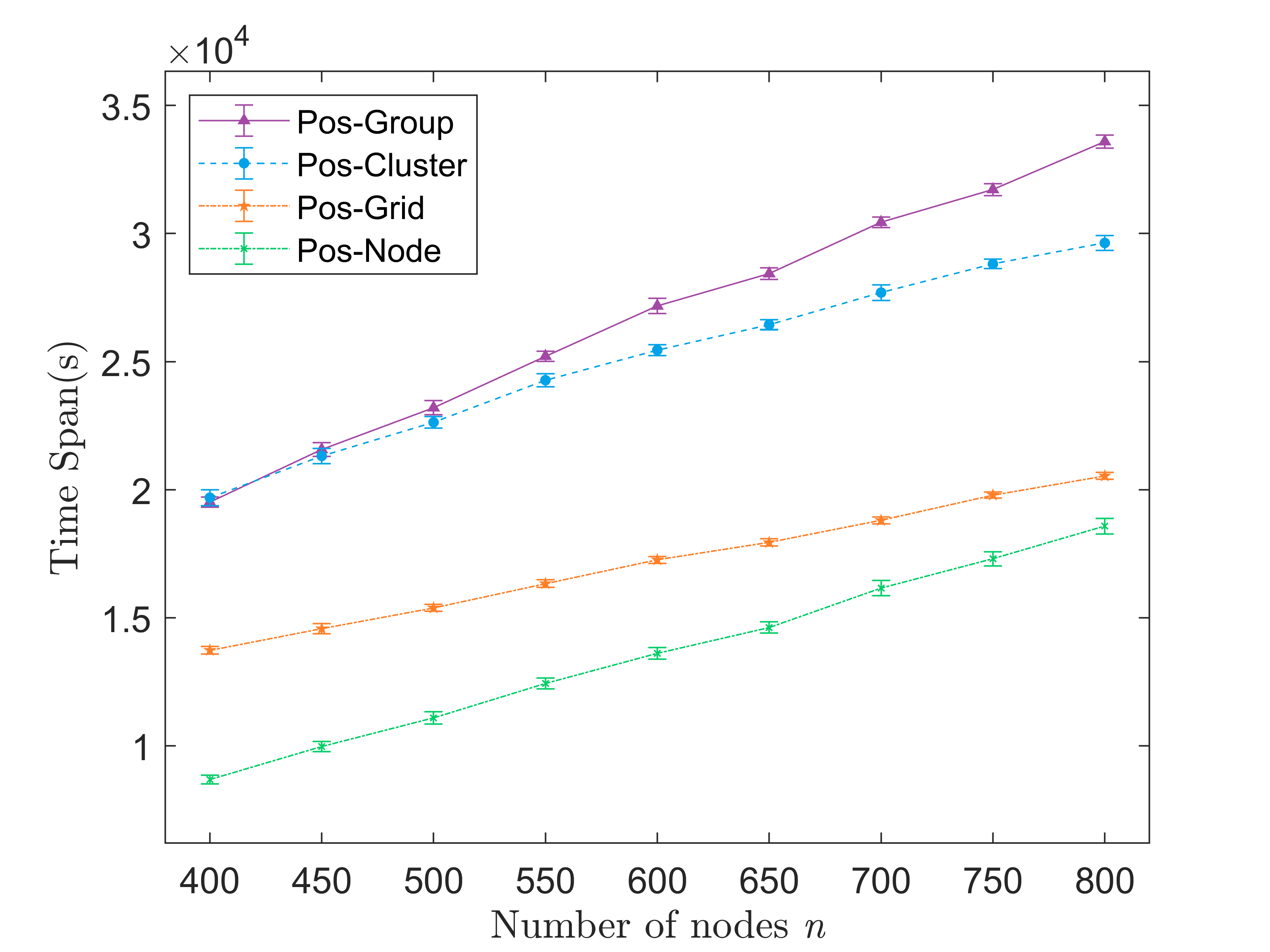}
        \caption{Time span}
        \label{fig_res1_time}
    \end{subfigure}
    \hfill
    \begin{subfigure}[b]{0.24\linewidth} % 调整宽度
        \centering
        \includegraphics[width=\linewidth]{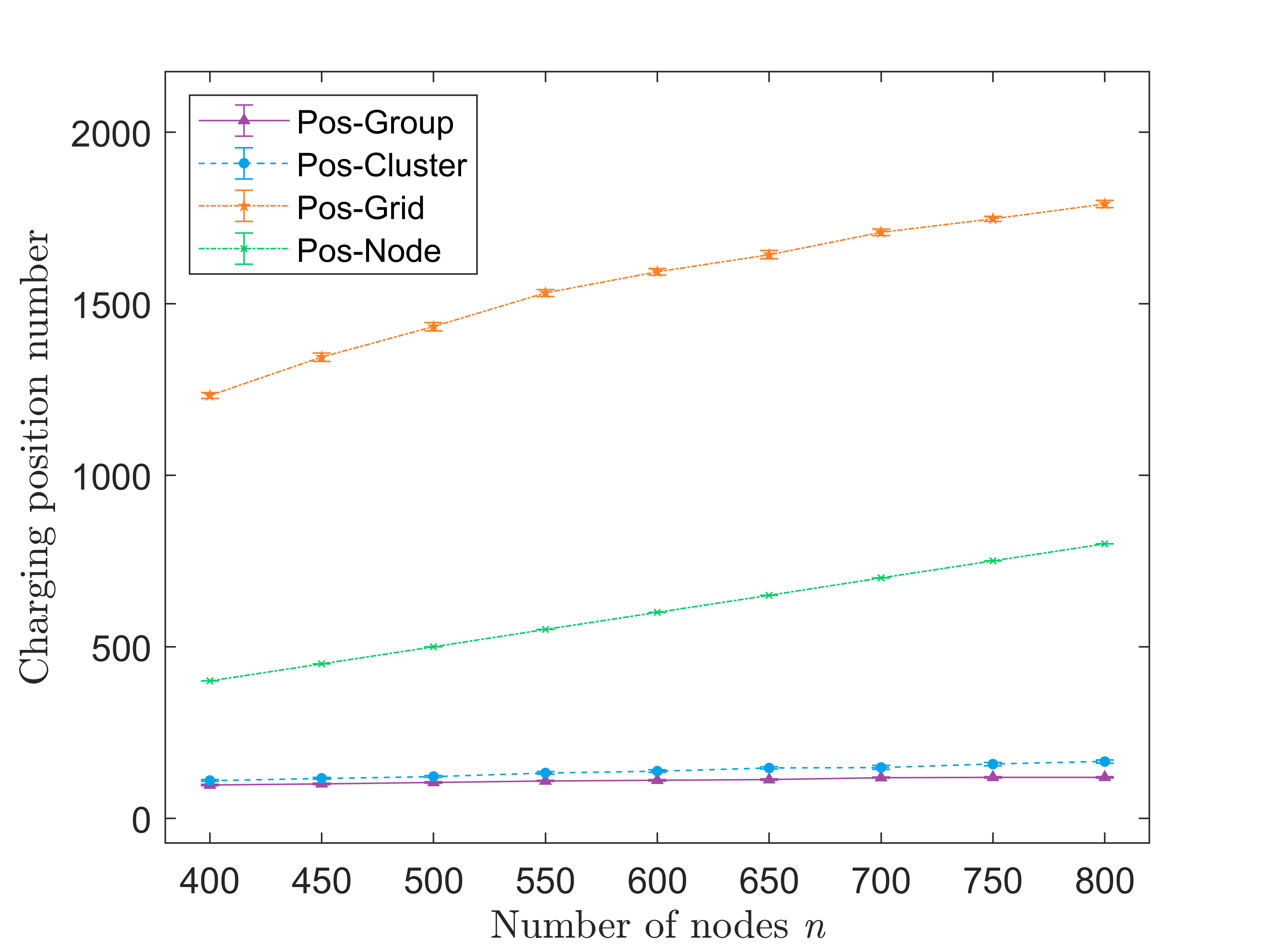}
        \caption{Charging position number}
        \label{fig_res1_count}
    \end{subfigure}
    \hfill
    \begin{subfigure}[b]{0.24\linewidth} % 调整宽度
        \centering
        \includegraphics[width=\linewidth]{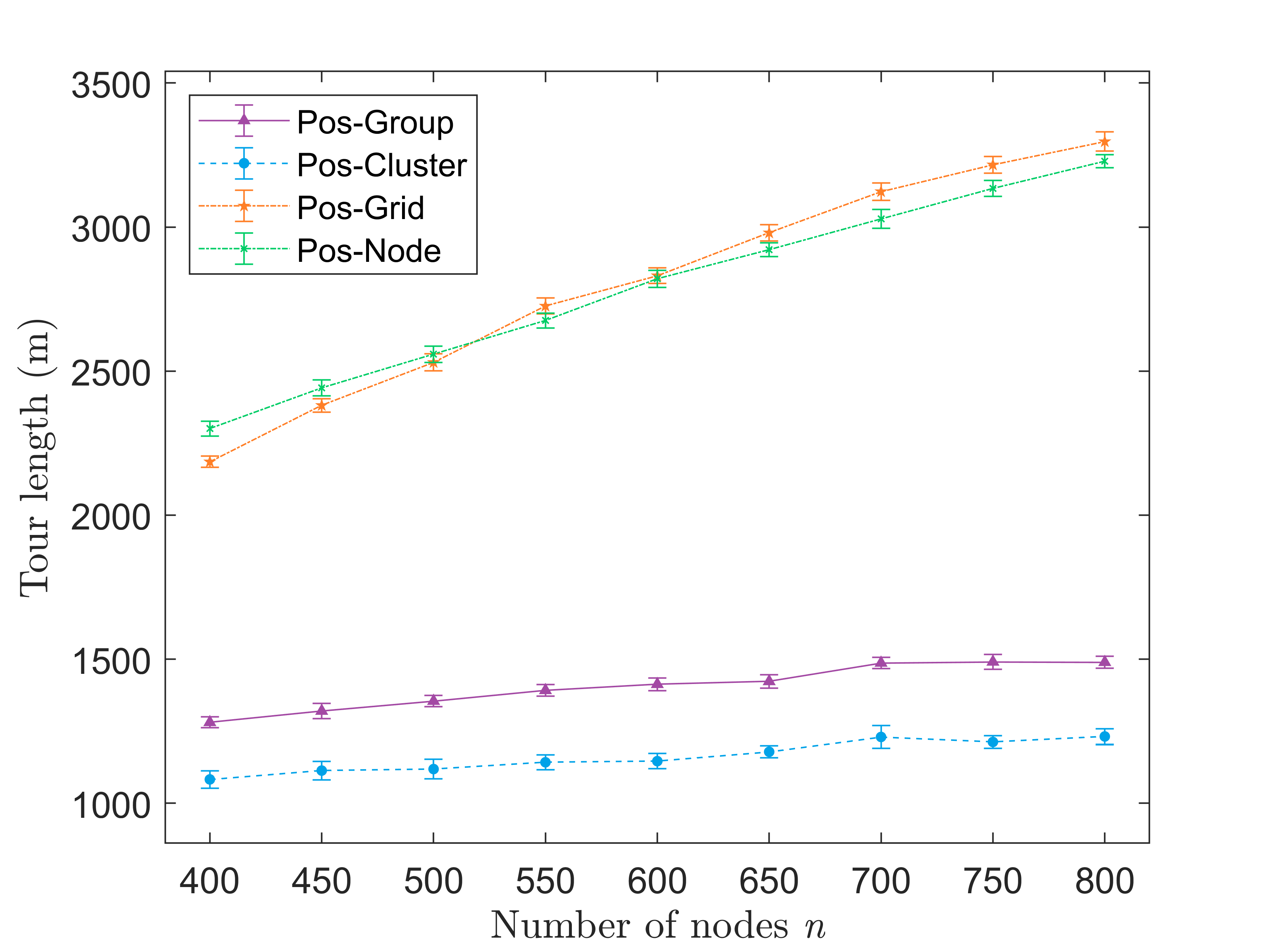}
        \caption{Tour length}
        \label{fig_res1_tour}
    \end{subfigure}
    \caption{Effects of the number of nodes}
    \label{fig_res1}
    \vspace{-0.6cm}
\end{figure*}

2) Ratio of $p_\text{Hov}/p_\text{Fly}$: The ratio $p_\text{Hov}/p_\text{Fly}$ increases from 0.1 to 1 with a step size of 0.1, while the number of nodes $n$ is fixed at $n=400$. The results are shown in Fig.~\ref{fig_res2}. The total energy loss $e^\text{Total}_\text{Loss}$ of the four algorithms all increase, indicating that the UAV's hovering energy consumption is a primary source of energy consumption of the charging schedules. It is also evident that $e^\text{Total}_\text{Loss}$ of Pos-Node consistently remains lower than those of the other three algorithms, confirming that Pos-Node is the best among the charging position set generation methods.

\begin{figure*}[!htb]
    \centering
    \begin{subfigure}[b]{0.3\linewidth}
        \centering
        \includegraphics[width=\linewidth]{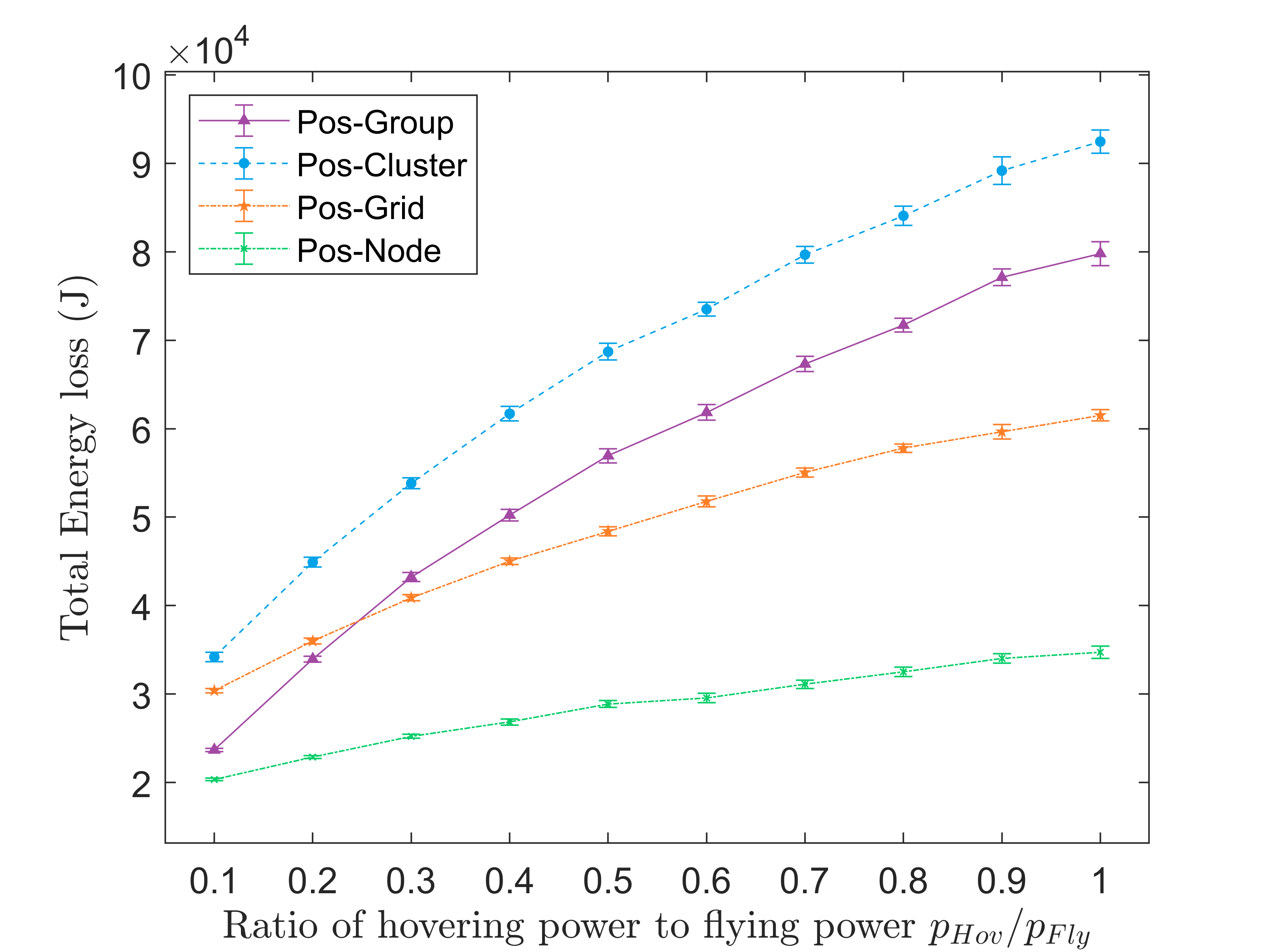}
        \caption{Total energy loss}
        \label{fig_res2_loss}
    \end{subfigure}
    \hfill
    \begin{subfigure}[b]{0.3\linewidth}
         \centering
        \includegraphics[width=\linewidth]{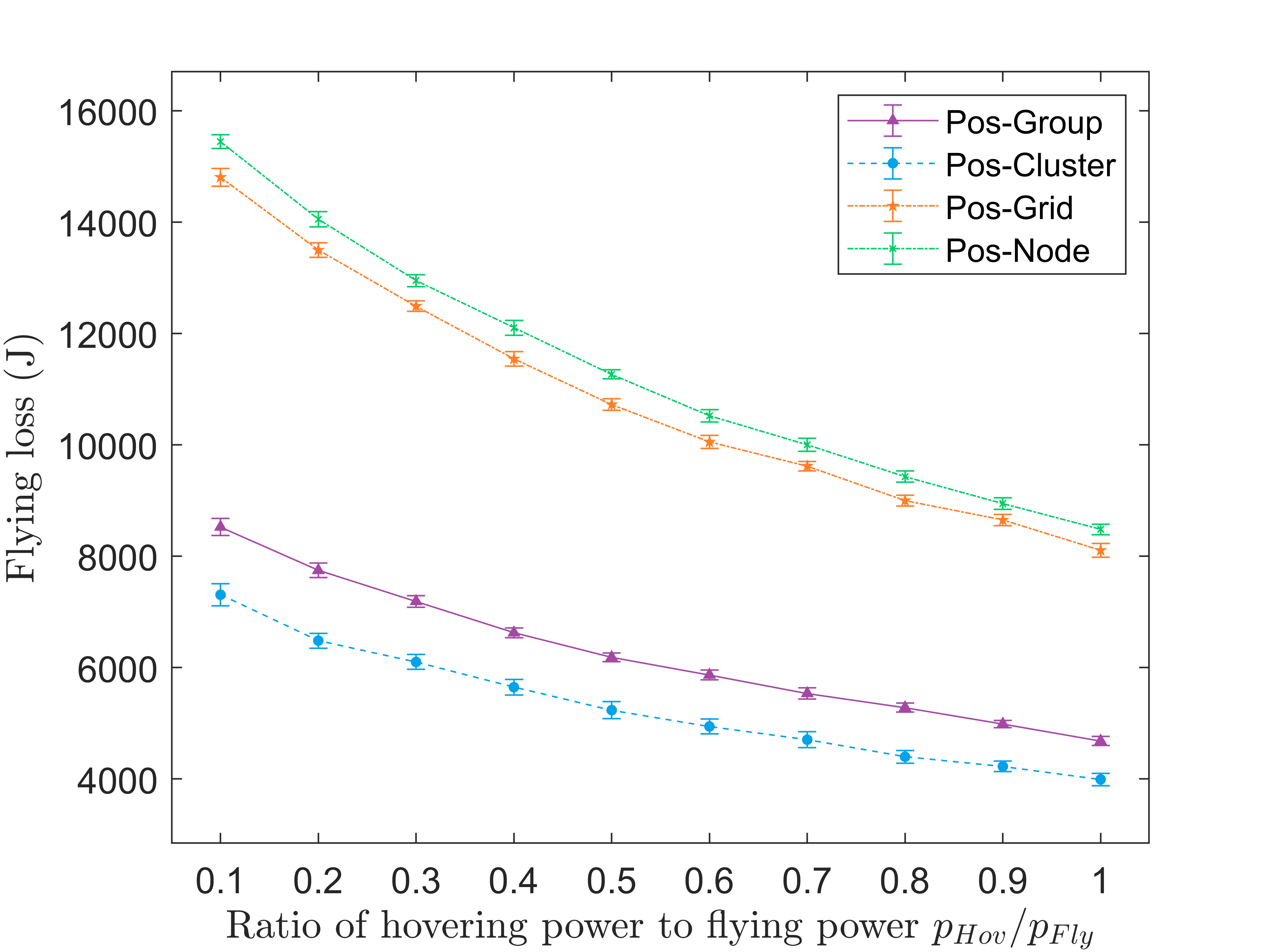}
        \caption{Flying energy consumption}
        \label{fig_res2_fly}
    \end{subfigure}
     \hfill
      \begin{subfigure}[b]{0.3\linewidth}
         \centering
        \includegraphics[width=\linewidth]{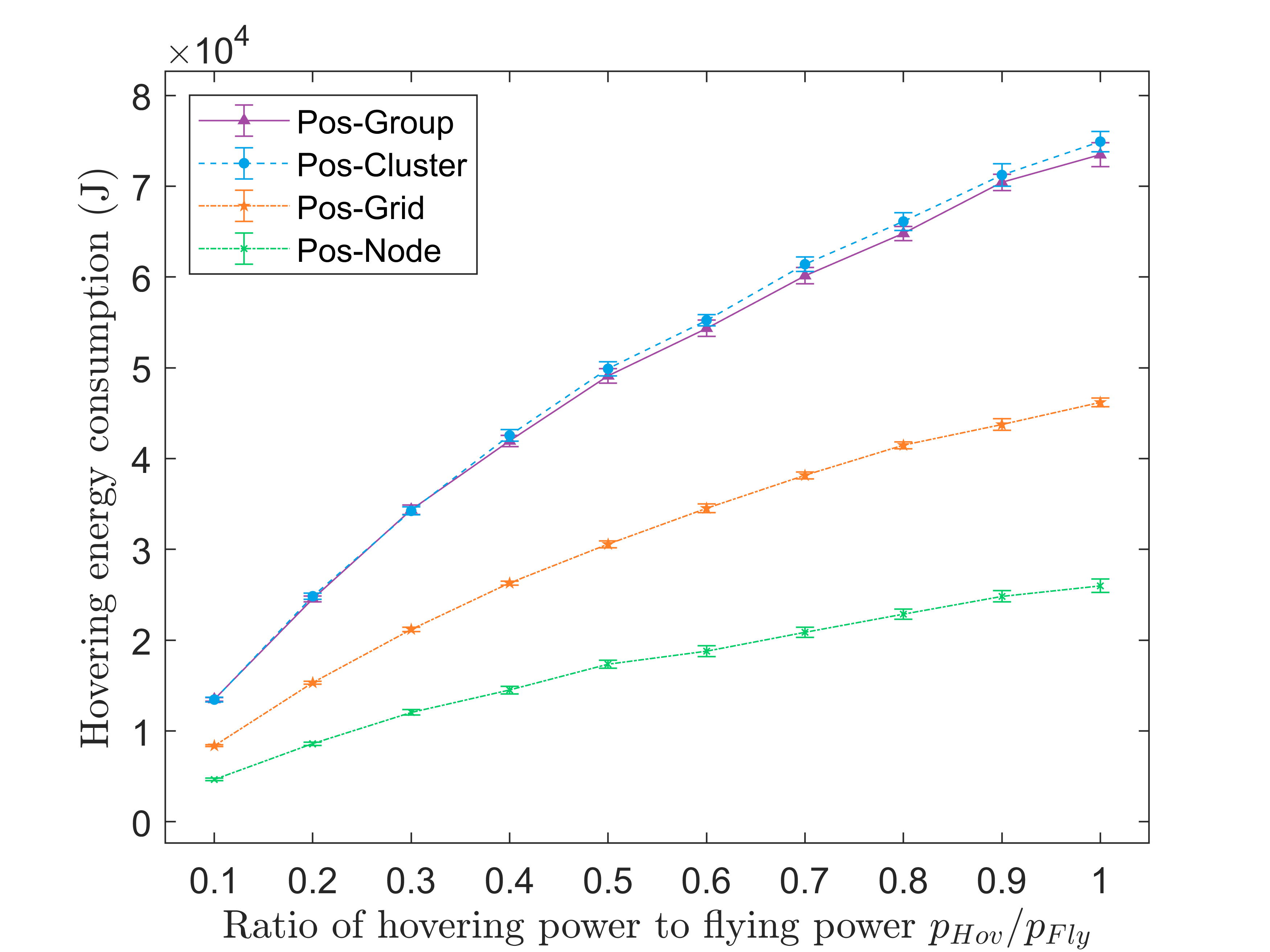}
        \caption{Hovering energy consumption}
        \label{fig_res2_hover}
    \end{subfigure}
    \caption{Effects of the ratio of hovering energy consumption power to flying energy consumption power.}
    \label{fig_res2}
    \vspace{-0.6cm}
\end{figure*}

% Results of detailed performance metrics for $n{=}400$ and $p_\text{Hov}/p_\text{Fly}{=}1$, such as  charging energy consumption and hovering energy consumption, are shown in Fig.~\ref{fig_res3}. Pos-Node outperforms the others in terms of the detailed metrics. Although Pos-Node generates more charging positions leading to larger flying energy consumption overhead, larger energy transfer coefficients considerably saves the charging time and leads to shorter time span, thus saves both charging and hovering energy consumption considerably.

\subsection{Comparison of Charging Direction Selection Methods}
\label{s7_performance_dc}
With the fixed options of Pos-Cluster and Tour-LKH, the performances of the five charging direction selection methods were evaluated with node number $n$ increases from 400 to 800. Six performance metrics are adopted in this work, including total energy loss, time span, charging direction number, direction choosing time, P3 solving time, and total running time. Direction choosing time metric measures the running time for choosing charging directions. P3 solving time represents the average time of the Cplex for solving P3, and total running time represents the time for running a whole scheme completely. The results are shown in Fig.~\ref{fig_res4}.

Fig.~\ref{fig_res4_loss} illustrates the total energy loss $e^\text{Total}_\text{Loss}$ of all the five methods all increase as $n$ increases. However, $e^\text{Total}_\text{Loss}$ of Dir-FuncEqv consistently remains lower than those of Dir-Node, Dir-ACC, and Dir-Polyhedron. This is because that Dir-Node, Dir-ACC, and Dir-Polyhedron fail to encompass all LMax-SCN set, which can be noticed easily with some analyses. Dir-GCC generates a charging direction set guaranteed to contains all LMax-SCN set, so it leads to schedules with total energy loss identical to Dir-FuncEqv. However, Dir-GCC's direction set is excessively redundant, as revealed in Fig.~\ref{fig_res4_count}, which makes Dir-GCC runs much slower than Dir-FuncEqv. Fig.~\ref{fig_res4_time} highlights the superior performance of Dir-FuncEqv in term of time span. Fig.~\ref{fig_res4_count} shows the number of charging directions determined by the methods, where the y-axis is in a logarithmic scale for the large scale differences of charging direction numbers. Fig.~\ref{fig_res4_dTime} shows that Dir-FuncEqv incurs higher time overhead, which is mainly due to an unfair point that Dir-FuncEqv contains a direction adjustment process for practical considerations, yet the others do not contains. In Fig.~\ref{fig_res4_cplex}, Dir-GCC consumes the most time due to the high redundancy in its charging directions. As shown in Fig.~\ref{fig_res4_rTime}, despite the additional time-consuming direction adjustment process, Dir-FuncEqv's total running time is not significantly different from the other algorithms. This demonstrates that Dir-FuncEqv is highly efficient.

\begin{figure}[!htbp]
    \centering
    % 第一行
    \begin{subfigure}[b]{0.48\linewidth} % 调整宽度
        \centering
        \includegraphics[width=\linewidth]{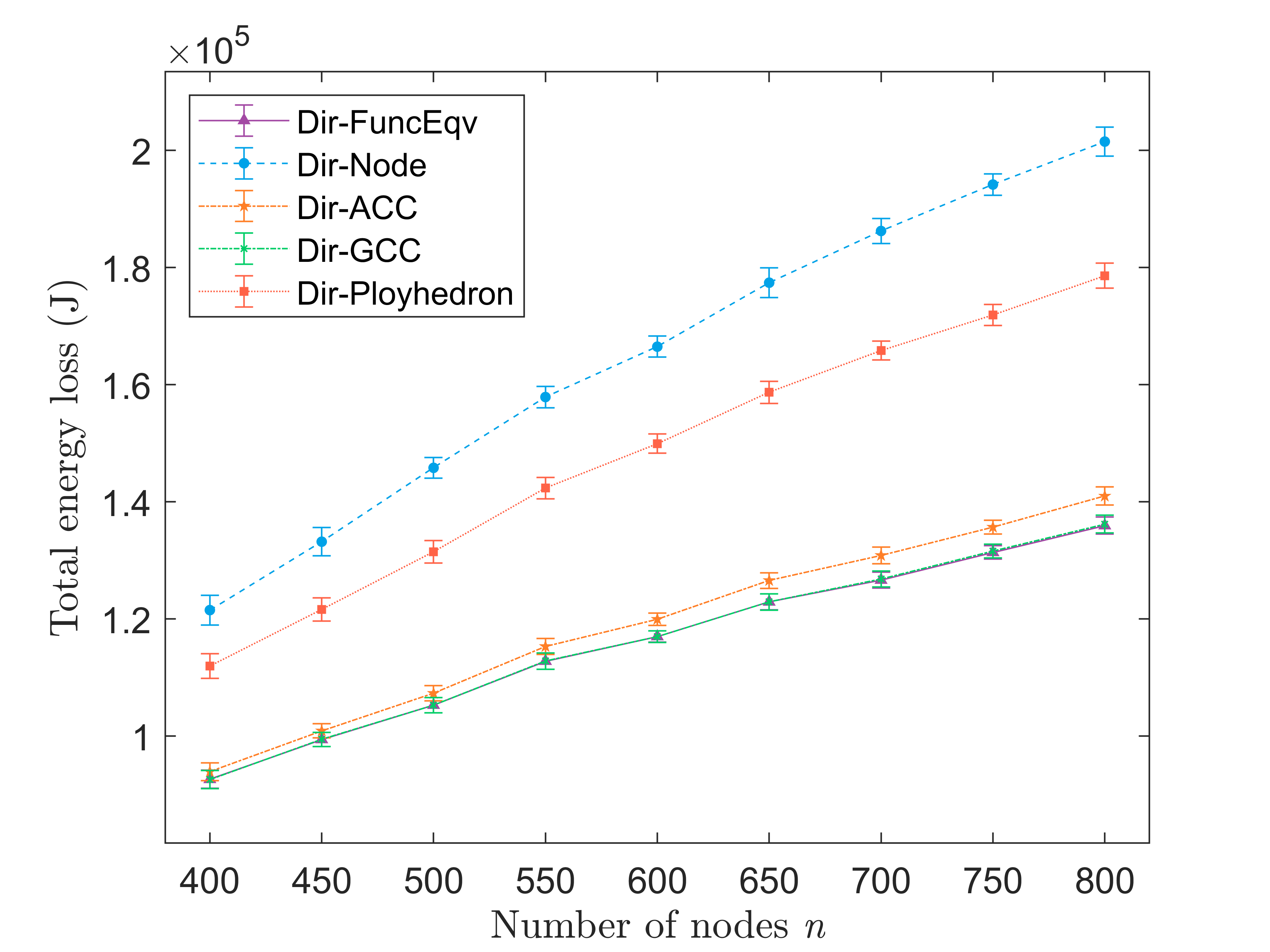}
        \caption{Total energy loss}
        \label{fig_res4_loss}
    \end{subfigure}
    \hfill
    \begin{subfigure}[b]{0.48\linewidth} % 调整宽度
        \centering
        \includegraphics[width=\linewidth]{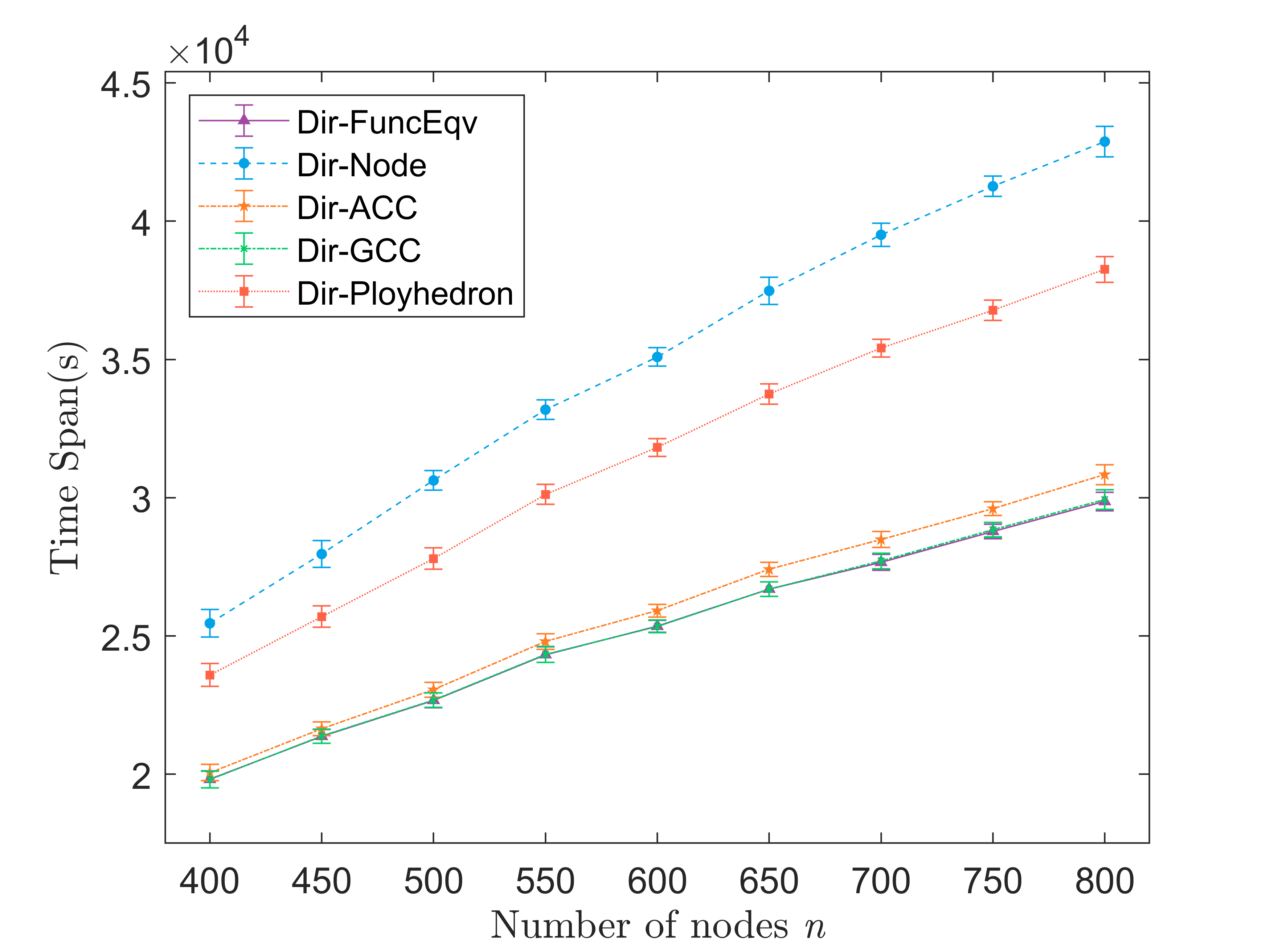}
        \caption{Time span}
        \label{fig_res4_time}
    \end{subfigure}
    
    % 第二行
    \begin{subfigure}[b]{0.48\linewidth} % 调整宽度
        \centering
        \includegraphics[width=\linewidth]{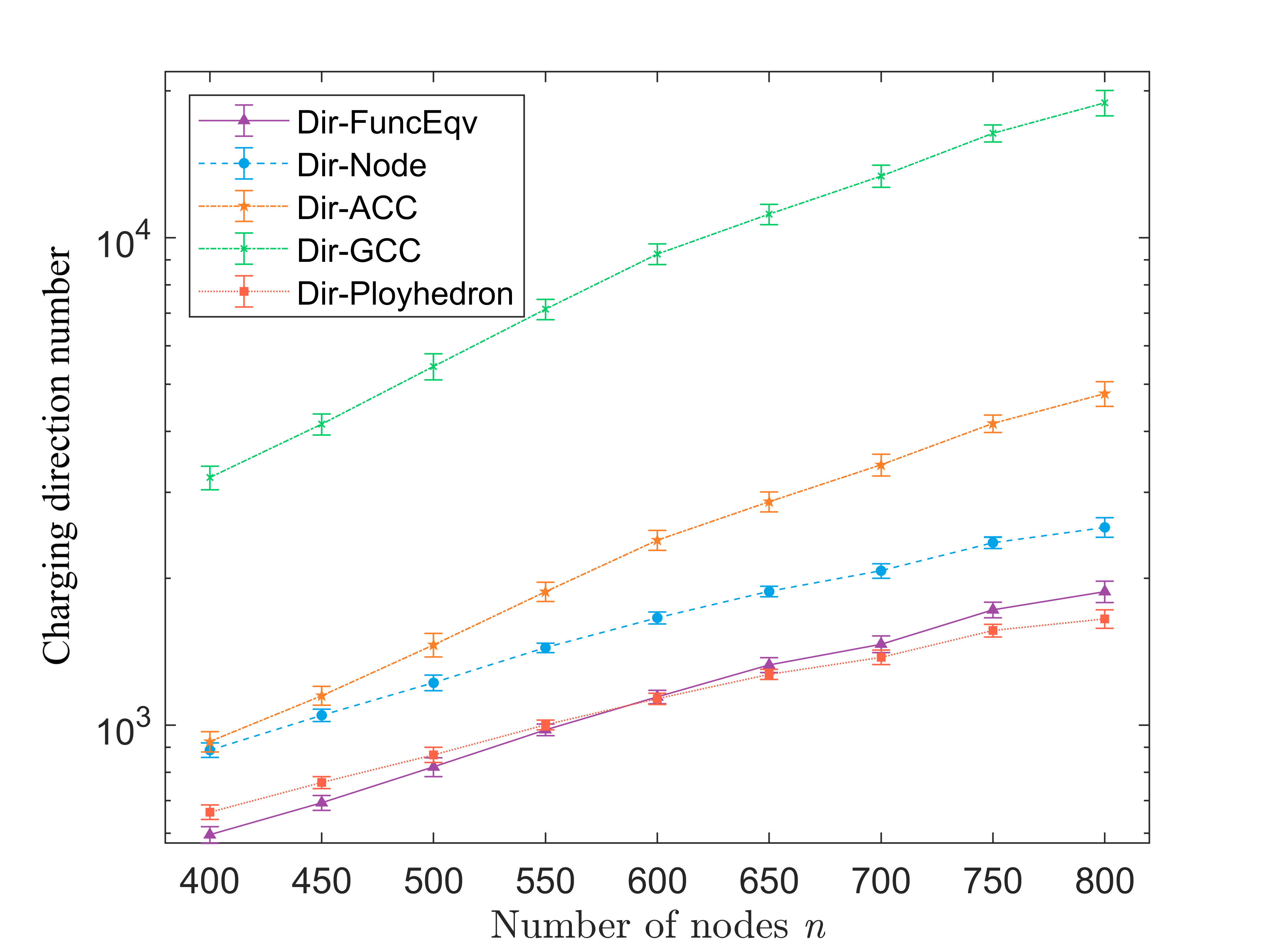}
        \caption{Charging direction number}
        \label{fig_res4_count}
    \end{subfigure}
    \hfill
    \begin{subfigure}[b]{0.48\linewidth} % 调整宽度
        \centering
        \includegraphics[width=\linewidth]{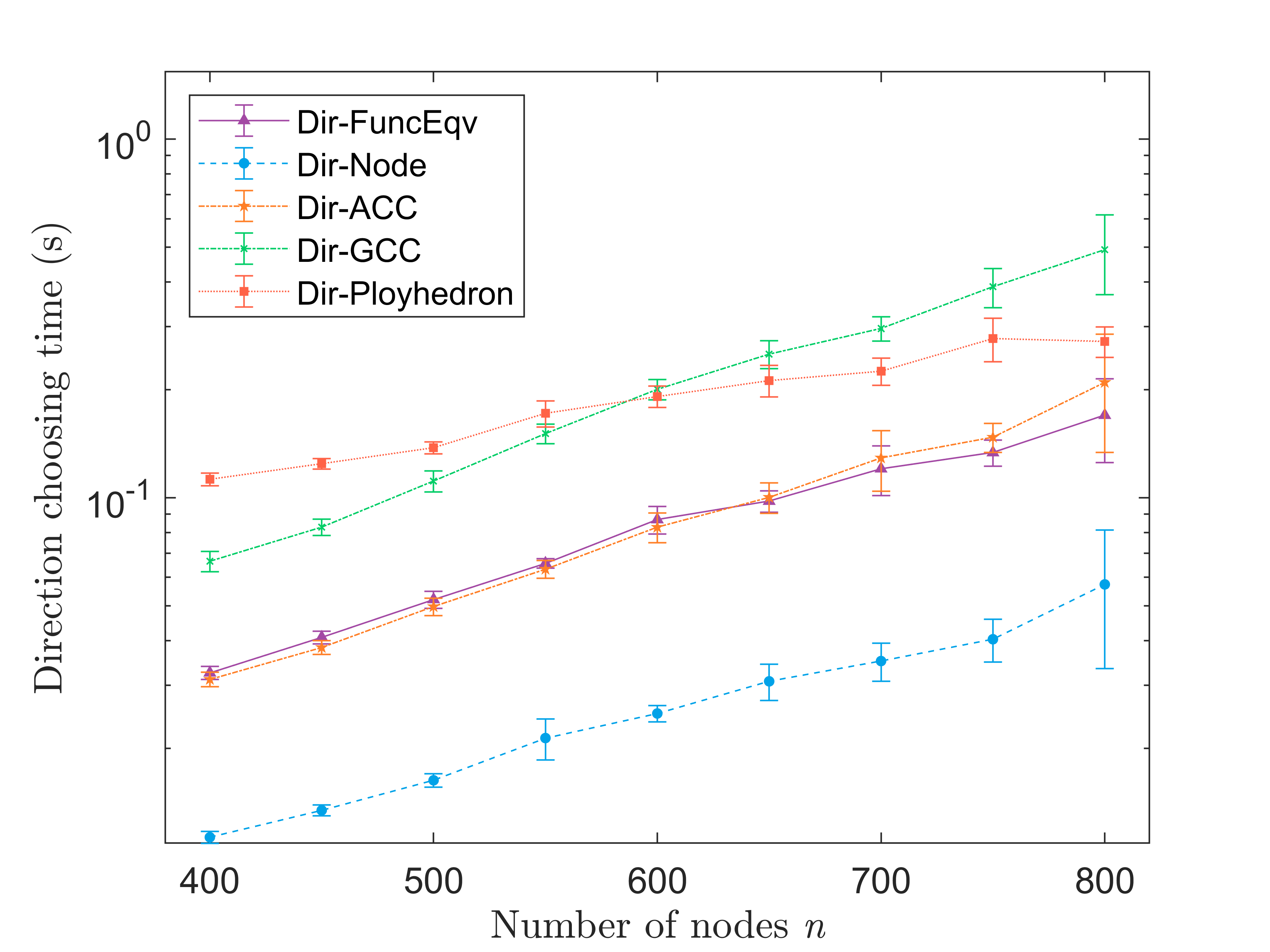}
        \caption{Direction choosing time}
        \label{fig_res4_dTime}
    \end{subfigure}
    
    % 第三行
    \begin{subfigure}[b]{0.48\linewidth} % 调整宽度
        \centering
        \includegraphics[width=\linewidth]{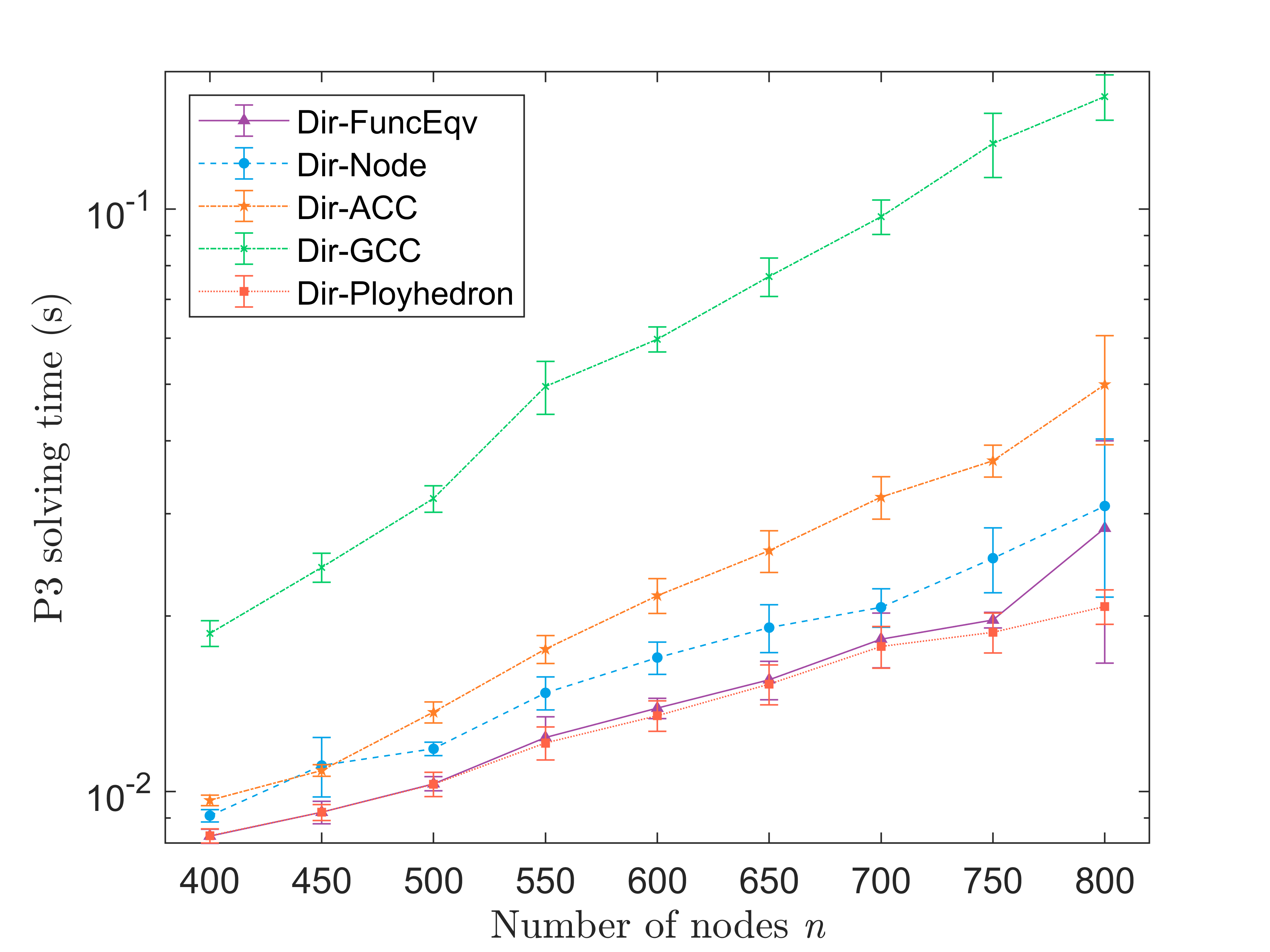}
        \caption{P3 solving time}
        \label{fig_res4_cplex}
    \end{subfigure}
    \hfill
    \begin{subfigure}[b]{0.48\linewidth} % 调整宽度
        \centering
        \includegraphics[width=\linewidth]{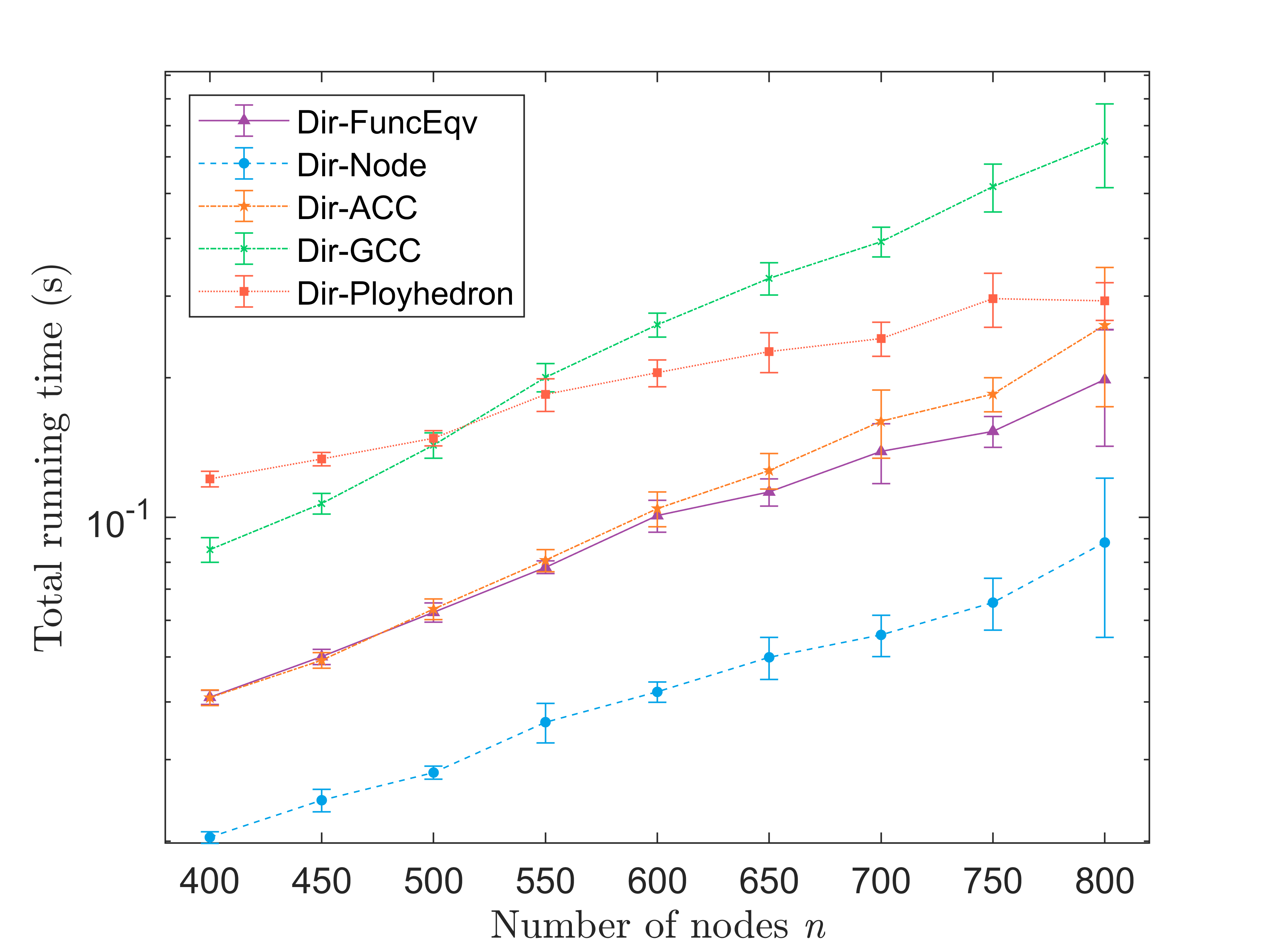}
        \caption{Total running time}
        \label{fig_res4_rTime}
    \end{subfigure}
    \caption{Performance comparison of charging direction selection methods}
    \label{fig_res4}
    \vspace{-0.5cm}
\end{figure}

\subsection{Comparison of Charging Tour Searching Methods}
To concentrate on the performance differences of \textit{Tour-LKH}, \textit{Tour-Ant}, and \textit{Tour-Greedy} in searching efficient charging tours, we fixed the options of Pos-Node and Dir-FuncEqv. The flying energy consumption and total running time performance metrics are shown in Fig.~\ref{fig_res5}.

As shown in Figure ~\ref{fig_res5_loss}, Tour-LKH consumes significantly less energy than Tour-Ant and Tour-Greedy. The results in Figure ~\ref{fig_res5_run} further demonstrate that Tour-LKH runs significantly faster than Tour-Ant. While Tour-Greedy boasts the fastest running speed due to the simplicity of its algorithm, this comes at the cost of significantly increased energy consumption.

\begin{figure}[!htb]
    \centering
    \begin{subfigure}[b]{0.48\linewidth}
        \centering
        \includegraphics[width=\linewidth]{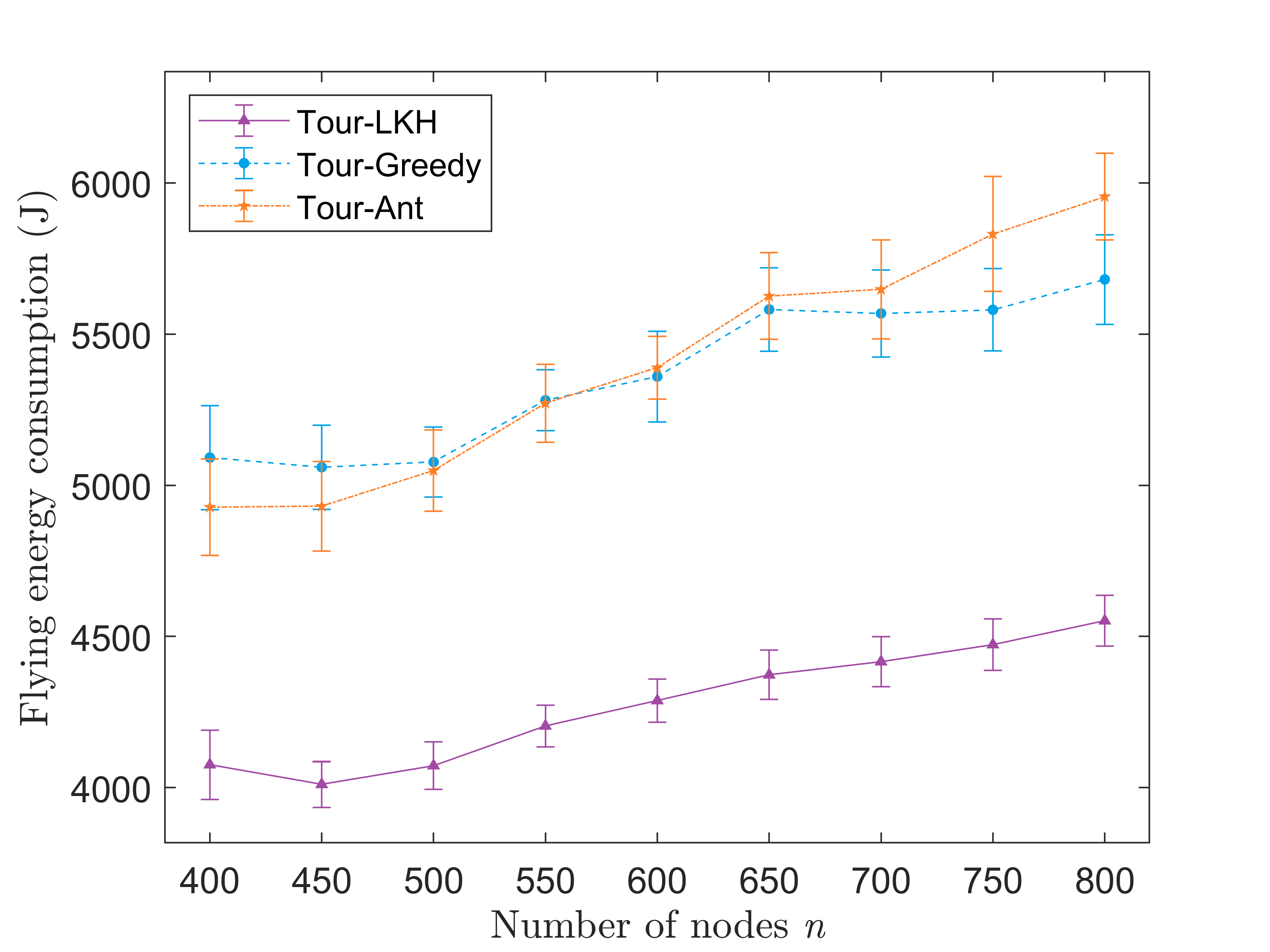}
        \caption{Flying energy consumption}
        \label{fig_res5_loss}
    \end{subfigure}
     \hfill
      \begin{subfigure}[b]{0.48\linewidth}
         \centering
        \includegraphics[width=\linewidth]{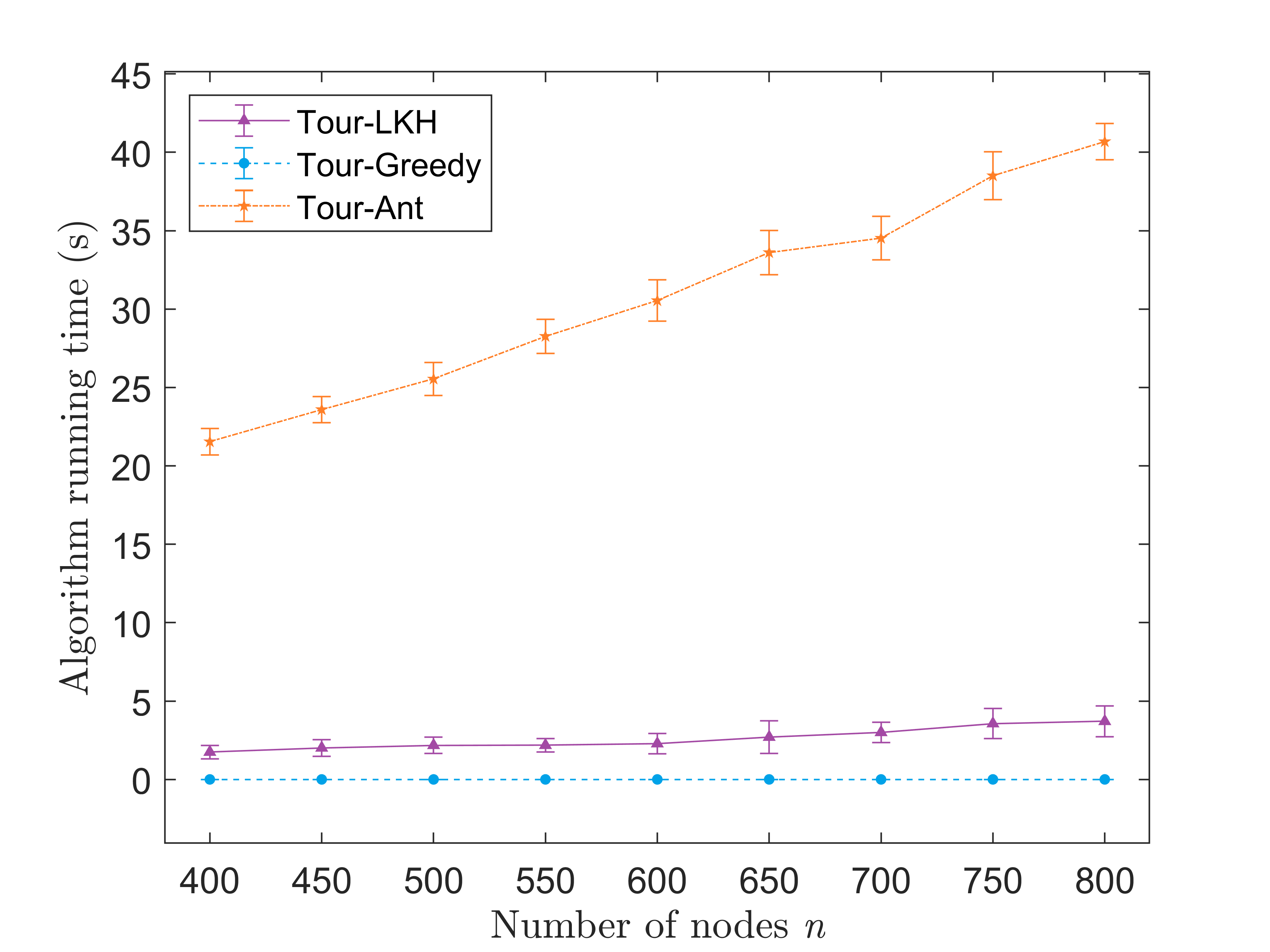}
        \caption{Total running time}
        \label{fig_res5_run}
    \end{subfigure}
    \caption{Performance comparison in solving TSP}
    \label{fig_res5}
    \vspace{-0.6cm}
\end{figure}

\subsection{Comparison of Complete Schemes}
\label{sec_cmp_complete_schemes}
Based on the results in the previous experiments, we select two additional complete schemes to compare with FELKH-3D. The first one is a combination of Pos-Group, Dir-Polyhedron, and Tour-Ant, which is denoted as ``\emph{Sch-GroupPolyAnt}`` and represents clustering based methods. Sch-GroupPolyAnt can be considered as a 3D extension of the algorithm proposed in \cite{Liang2023TVT}. The second one combines Pos-Grid, Dir-ACC, and Tour-Greedy, which is denoted as ``\emph{Sch-GridAccGreedy}`` and represents traditional methods. Sch-GridAccGreedy can be regarded as an adaption of the algorithm proposed in \cite{Jiang}. 

The results are shown in Fig.~\ref{fig_res6}. Fig.~\ref{fig_res6_loss} shows the results of the total energy loss metric. As FELKH-3D combines all the most energy efficient options in the steps, it consistently exhibits lower energy loss compared to the other schemes. As shown in Fig.~\ref{fig_res6_time}, FELKH-3D also produces the optimal charging schedules in term of time span.

\begin{figure}[!htb]
    \centering
    \begin{subfigure}[b]{0.48\linewidth}
        \centering
        \includegraphics[width=\linewidth]{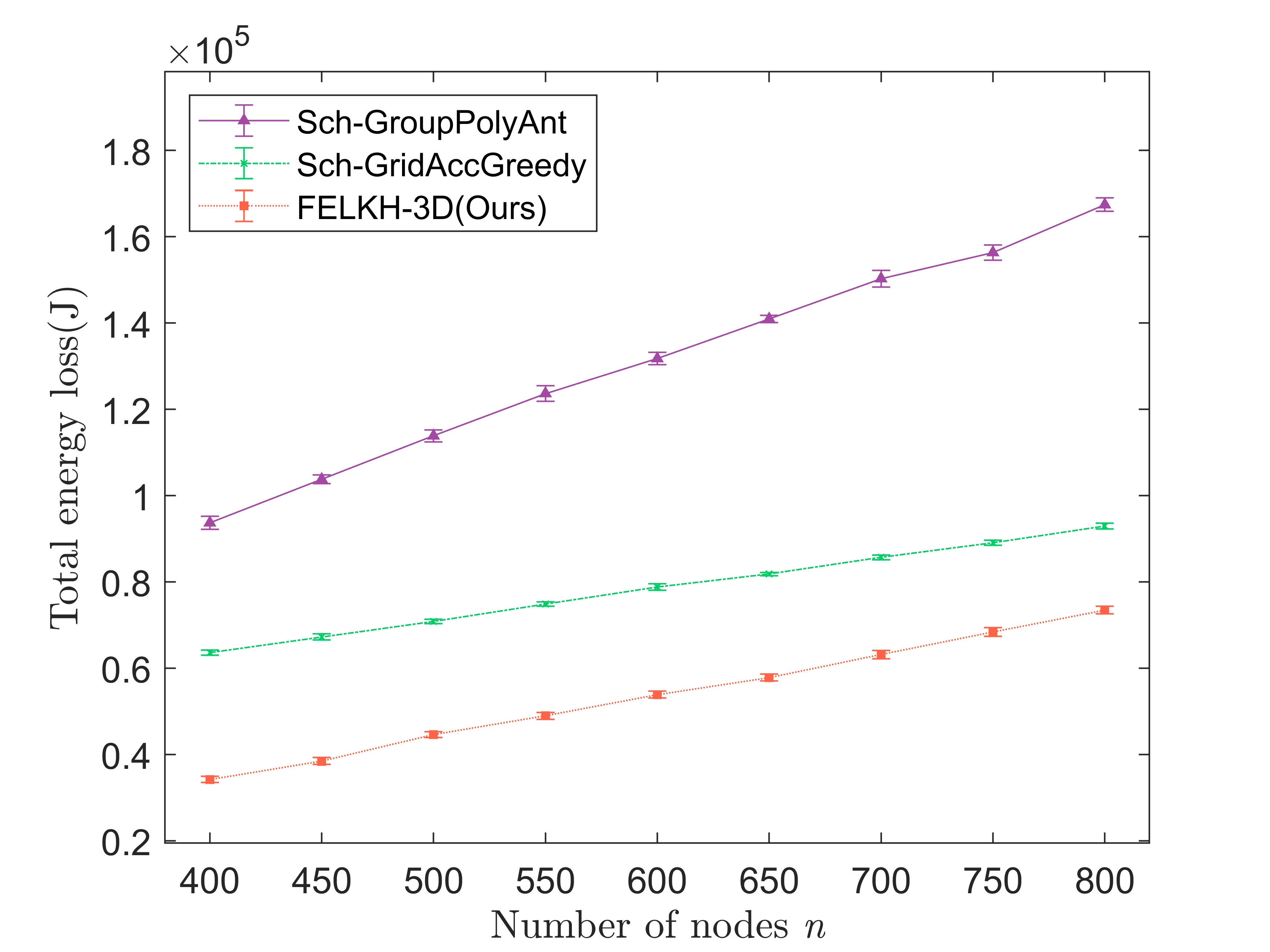}
        \caption{Total energy loss}
        \label{fig_res6_loss}
    \end{subfigure}
    \hspace{\fill}
    \begin{subfigure}[b]{0.48\linewidth}
         \centering
        \includegraphics[width=\linewidth]{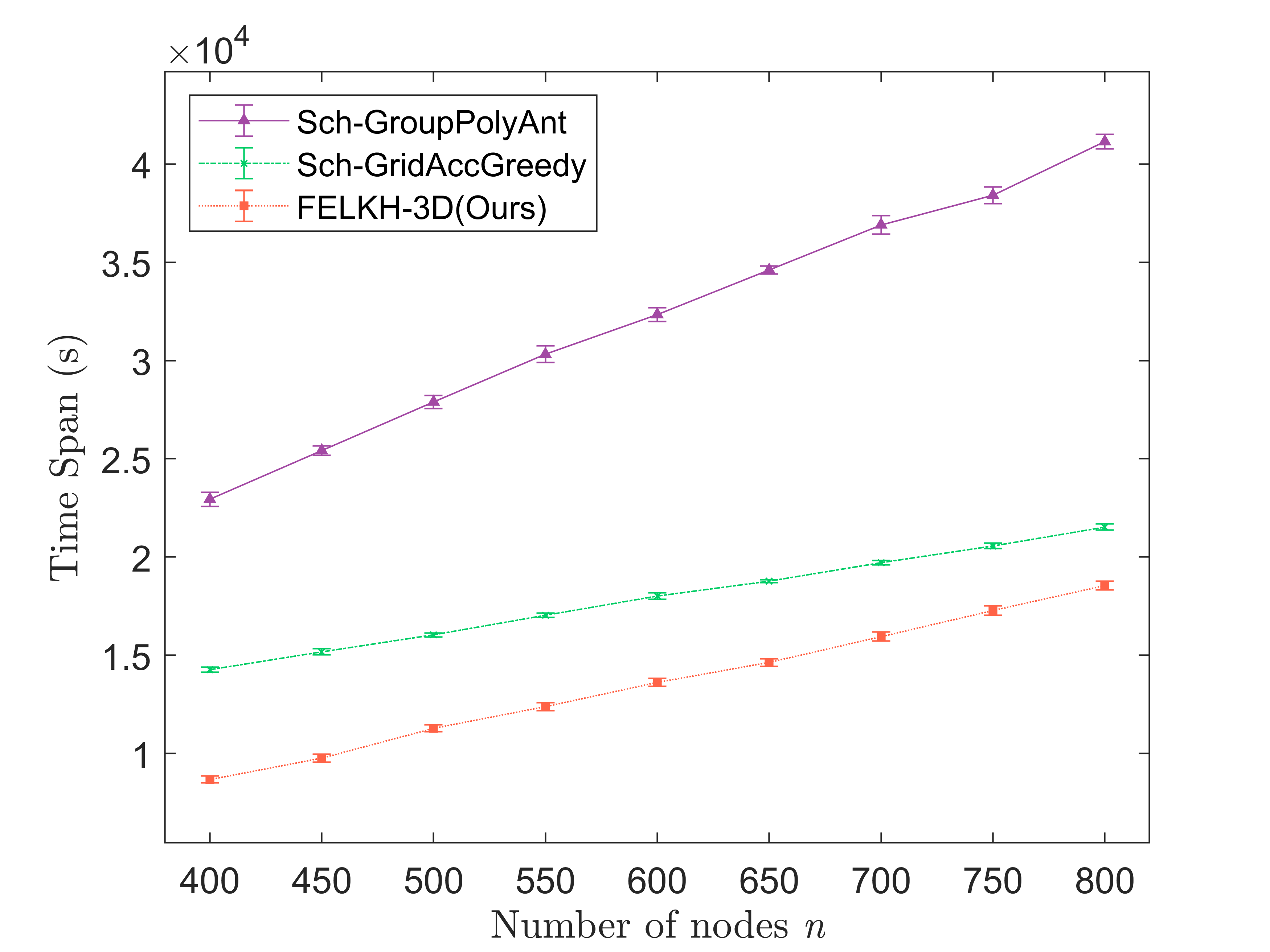}
        \caption{Time span}
        \label{fig_res6_time}
    \end{subfigure}
    % \hspace{\fill}
    % \begin{subfigure}[b]{0.3\linewidth}
    %      \centering
    %     \includegraphics[width=\linewidth]{images/EvaFig/fig_entire_runTime.png}
    %     \caption{Algorithm running time}
    %     \label{fig_res6_runTime}
    % \end{subfigure}
    \caption{Performance comparison of complete schemes for solving DCS-3D problems}
    \label{fig_res6}
    \vspace{-0.5cm}
\end{figure}

\section{ Conclusion}
\label{sec_conclusion}
In this paper, we addressed the DCS-3D problem for 3D-WRSNs using UAVs as wireless chargers, focusing on scenarios where nodes are distributed in three-dimensional space. We established that the DCS-3D problem is NP-hard and introduced the FELKH-3D algorithm to solve it. To determine an optimal set of charging directions, we developed the cMFEDS algorithm. cMFEDS identifies a minimal set of charging directions in a three-dimensional space such that the set is functional equivalent to the original infinite directions forming the whole sphere surface, which is essential for producing optimal charging schedules. 
Thus, the challenge of infinite charging direction space is address by the cMFEDS algorithm. To determine the optimal charging tour for the UAV, the LKH algorithm is employed. Simulation results demonstrated that FELKH-3D achieves the best performance.

\small
\bibliographystyle{IEEEtran}
\bibliography{IEEEabrv,reference}

\begin{IEEEbiography}
[{\includegraphics[width=1in,height=1.25in,keepaspectratio]{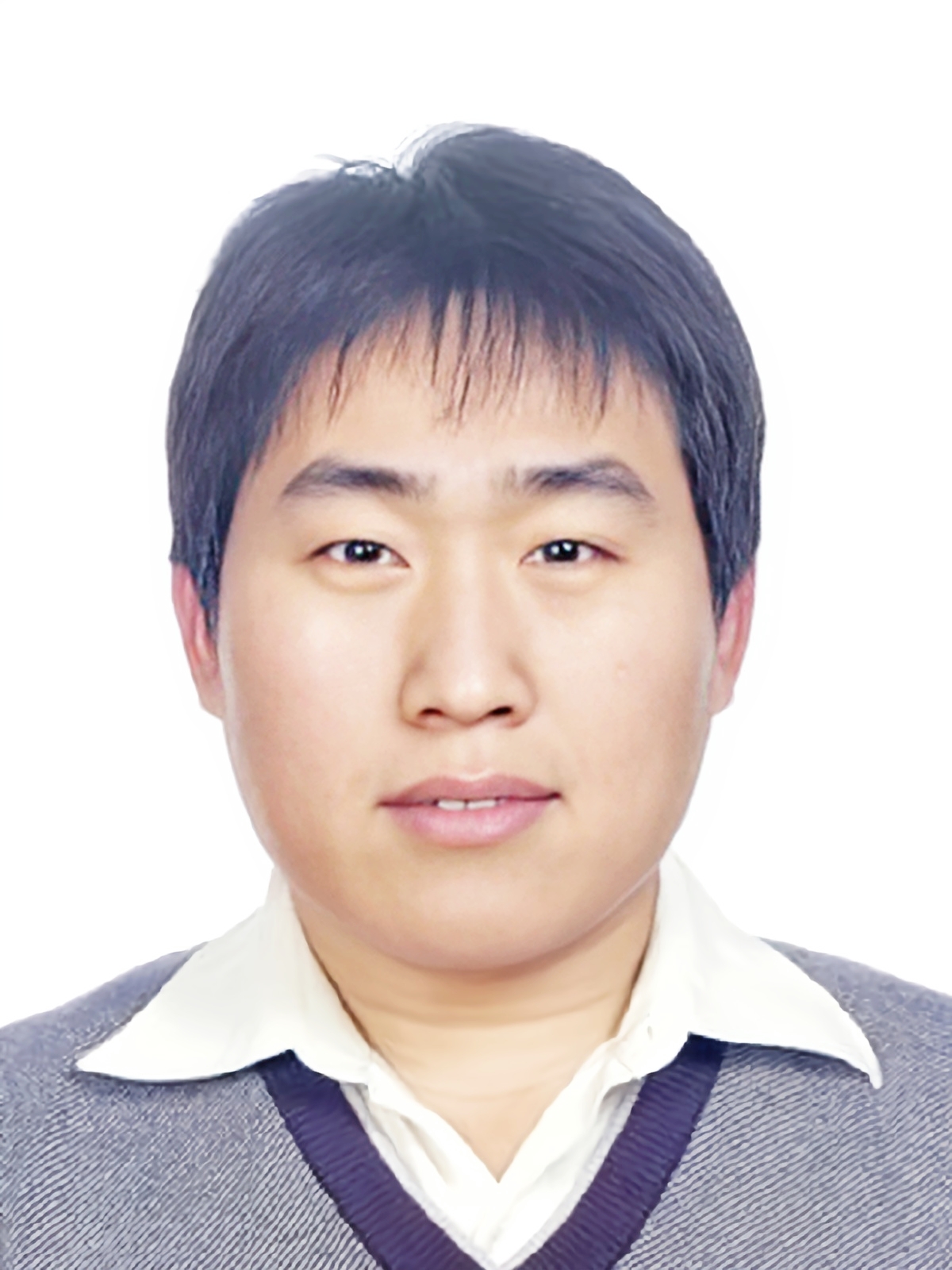}}]
{Zhenguo Gao} (Senior Member, IEEE) received the B.S. and M.S. degrees in mechanical and electrical engineering and the Ph.D. degree in computer architecture from Harbin Institute of Technology, Harbin, China, in 1999, 2001, and 2006, respectively.

He is currently a Professor with Huaqiao University, Xiamen, China, and also the Dean of the Key Laboratory of Computer Vision and Machine Learning (Huaqiao University), Fujian Province University, Xiamen. His research interests include wireless networks and edge computing.
\end{IEEEbiography}

\end{document}